\documentclass[11pt,a4paper]{article}

\usepackage{amsmath}
\usepackage{amssymb}
\usepackage{amsthm}
\usepackage{amsfonts}
\usepackage{mathrsfs}
\usepackage{mathtools}
\usepackage{physics}
\usepackage{graphicx}
\usepackage[margin=2.5cm]{geometry}
\usepackage{hyperref}
\usepackage{url}

\newtheorem{theorem}{Theorem}[section]
\newtheorem{lemma}[theorem]{Lemma}
\newtheorem{proposition}[theorem]{Proposition}

\theoremstyle{definition}
\newtheorem{definition}[theorem]{Definition}
\newtheorem{axiom}[theorem]{Axiom}
\newtheorem{postulate}[theorem]{Postulate}
\theoremstyle{remark}
\newtheorem{remark}[theorem]{Remark}

\title{\textbf{A Jordan Algebraic Approach to the Landau Problem with a Non-commutative
Parameter: \\ From Non-commutativity to Non-associativity}}

\author{
Tekin Dereli$^{\,1,2,*}$ \quad Ekin S\i la Y\"or\"uk$^{\,1,\dagger}$\\[6pt]
\small (1) Department of Physics, Ko\c{c} University, 34450 Sar\i yer, Istanbul, Turkey\\[2pt]
\small (2) Faculty of Aviation and Aeronautical Sciences,\\
\small \"Ozye\u{g}in University, 34794 \c{C}ekmek\"oy, Istanbul, Turkey\\[2pt]
\small $^{*}$\texttt{tdereli@ku.edu.tr}\quad
$^{\dagger}$\texttt{eyoruk13@ku.edu.tr}
}

\date{Revised version: October 6, 2026}

\begin{document}
\maketitle

\begin{abstract}
Non-associative structures arise naturally in several areas of fundamental physics, including
the quantum mechanics of charged particles in magnetic monopole backgrounds, or non-geometric
flux compactifications in string theory. A key observation motivating the present work is that
the non-commutativity of spatial coordinates can be recast as the non-associativity of
observables within a Jordan algebraic framework. This paper provides a comprehensive review of
the Jordan algebraic, non-associative formulation of quantum mechanics and demonstrates this
passage from non-commutativity to non-associativity concretely through the non-commutative
Landau problem coupled to a harmonic potential. We begin with a self-contained account of the
operator algebra background, covering $C^*$-algebras, von Neumann algebras, and Jordan--Banach
($JB$) and Jordan--von Neumann ($JBW$) algebras, and present the axiomatic formulation of
quantum mechanics in terms of the Jordan algebra of observables. We then develop an
alternative Hilbert space construction based on 3-cyclic traces, extending the work of Schupp and Szabo to infinite dimensional Jordan algebras. A central result of this
work is that the non-commutative parameter $\theta$ can be expressed entirely in terms of an
associator in the Jordan algebraic setting. Pure product states and density matrices arising from this construction are characterized within the JBW-algebra framework. A key advantage of our
approach is the \emph{explicit} construction of split operators for the non-commutative Landau
Hamiltonian, turning an abstract existence result into a concrete decomposition. This leads to
the Jordan--Schr\"odinger time-evolution equation for state vectors, whose exact closed-form
solution is presented. Since the results are obtained from the associative product of
the related $C^*$-algebra, they may also be useful for other non-associative formulations, such
as those arising in magnetic monopole dynamics, R-flux string backgrounds, and Lie triple
algebra settings, which we leave for future work.
\end{abstract}

\medskip
\noindent\textbf{Keywords:} non-associative algebras, Jordan algebras, algebraic
formulation of quantum mechanics, Landau problem, non-commutative geometry, split operators

\newpage

\section{Introduction}

\bigskip

In the algebraic formulation of quantum mechanics, observables play a central role. The main
motivation for this approach is to formulate quantum mechanics in terms of the essential
ingredients only; which are observables, states, expectation values and their time evolution.
In this approach, which is initiated by Segal \cite{Segal}, observables are taken as the
self-adjoint part of a $C^*$-algebra $\mathcal{A}$, states are positive linear functionals
on $\mathcal{A}$, and the real number that a state assigns to an observable is the expectation
value of that observable in that state.\footnote{See \cite{fano}, \cite{sorkin}, \cite{kempf}
for more abstract formulations of quantum mechanics.} However, it turns out that the
$C^*$-algebra one uses in such approaches usually allows for non-self-adjoint entities.
Consequently, a more restrictive structure is necessary to have a consistent description of the
algebra of observables to start with.

\medskip

The Jordan algebraic formulation \cite{JNW} adopted here is based on the fact that quantum
mechanical observables form an algebra $\mathcal{A}$ under the non-associative Jordan product:
$$
A\bullet B=\frac{1}{2}(AB+BA).
$$
This algebra qualifies as a real Jordan algebra of operators since the Jordan product is
commutative and satisfies the Jordan identity
$$
(A\bullet B)\bullet A^2=A\bullet (B\bullet A^2),
$$
which is equivalent to power-associativity of $\mathcal{A}$ when $\mathcal{A}$ is formally
real \cite{JNW}. Jordan algebras were originally introduced as a mathematical model of
quantum mechanics motivated by the following observation \cite{JNW}. Since the Jordan product
can be expressed by squares only,
$$
A\bullet B=\frac{1}{4}\bigl((A+B)^2-(A-B)^2\bigr),
$$
the set of observables is closed under Jordan multiplication. If one wants to define a
multiplication for observables, one therefore has to use the Jordan product and work with the
corresponding Jordan algebra. A Jordan algebra is called \emph{special} if it comes from the
Jordan product in an associative algebra, which is the case in our definition of the Jordan
algebra of observables.

\medskip

Since the algebras arising in quantum mechanics are typically infinite dimensional, it is of
interest to develop a structure theory compatible with the infinite dimensional case. This idea
is initiated by \cite{ASS}, where the direct analogues of $C^*$-algebras, called $JB$-algebras,
are studied. $JB$-algebras are Jordan algebras with a Banach space structure. They satisfy the
following conditions:
\begin{align}
&\|A^2\|=\|A\|^2, \nonumber\\
&\|A^2-B^2\|\leq \max\{\|A^2\|,\|B^2\|\}, \nonumber\\
&\|A\bullet B\|\leq \|A\|\|B\|. \nonumber
\end{align}
Here the first condition is the Jordan analogue of the $C^*$ condition, the second is a
positivity condition and the third is the analogue of the Banach algebra axiom; the standard
form of the axioms is recalled in Section~2.3. In particular, the observables, which can be considered as the self-adjoint part of a $C^*$-algebra, form a
$JC$-algebra, or in other words, a special $JB$-algebra. At this point, we can say that if
one wants to formulate quantum mechanical problems using the essential ingredients only, it is
important to write down a formulation using the Jordan algebraic approach. Moreover, every $JB$-algebra possesses a purely exceptional ideal whose quotient is a
$JC$-algebra, i.e.\ a norm-closed Jordan algebra of self-adjoint operators on a complex Hilbert
space; the exceptional part is built from the Albert algebra $H_3(\mathbb{O})$ \cite{ASS}. We refer to \cite{townsend}, \cite{niestegge},
\cite{Gunaydin}, \cite{gunaydin2} for more detailed work on the exceptional character of
$H_3(\mathbb{O})$ and Jordan algebras.

\medskip

Apart from Jordan algebras, we should also note other non-associative structures which have
found applications in physical problems \cite{gunaydinandzumino}, \cite{jackiw1}, \cite{szabo1},
\cite{szabo2}, \cite{mickelsson}, \cite{bojowald}. The associative product used in quantum
mechanics does not apply to some exotic systems such as the problem of a charged particle with a
magnetic monopole background, which leads to non-associative algebras. In such problems, a
non-associative extension of quantum mechanics becomes inevitable. A version of quantum mechanics
that can handle non-associative algebras of observables is given in \cite{SchuppSzabo}. This
approach has the merit of being naturally probabilistic. The algebra used to describe magnetic
monopoles is known not to possess a Jordan algebraic structure \cite{B}. However, some of our
results, related to explicit decompositions of unbounded operators into finite sums of commutators,
are obtained independently of the Jordan algebraic structure. Therefore, our method has a certain flexibility and may also be useful for other
non-associative formulations. We further extend the development of \cite{SchuppSzabo}
to the infinite dimensional Jordan algebra of observables. Similar to \cite{SchuppSzabo}, our
development will be based on a GNS-type construction and 3-cyclic traces (see also
\cite{MylonasSchuppSzabo} as the original source for the use of 3-cyclic traces in this context).

\medskip

It turns out that if we consider the dynamics of electrons in uniform distributions of magnetic
charge, certain non-associative deformations of quantum mechanics and gravity in three dimensions
are needed. The quantitative framework for non-associative quantum mechanics in this setting
exhibits new effects compared to ordinary quantum mechanics with sourceless magnetic fields
\cite{bojowald}. On the other hand, the non-associativity of translations in a quantum system
with a background magnetic field has received renewed interest in association with topologically
trivial gerbes over $\mathbb{R}^n$. The non-associativity arises when trying to lift the
translation group action on the 1-particle system to a second quantized system, and it is
described by a 3-cocycle of the group $\mathbb{R}^n$ with values in the unit circle $S^1$.

\medskip

Non-associative structures also emerge prominently in string theory. Starting from a torus $T^3$
with constant $H$-flux and applying successive T-duality transformations leads to a chain
$$
H_{ijk}\;\xrightarrow{\;T_i\;}\; f^{i}{}_{jk}\;\xrightarrow{\;T_j\;}\;Q^{ij}{}_k
\;\xrightarrow{\;T_k\;}\;R^{ijk},
$$
where the final non-geometric \emph{R-flux} background admits no conventional spacetime
description \cite{BakasLust}. In this R-flux string model the phase-space star-product is
non-associative, and the failure of associativity is controlled by a 3-cocycle of the Abelian
group of translations in double phase space \cite{BakasLust}, \cite{MylonasSchuppSzabo}. The
magnetic field analogue of the R-flux model is provided by a constant uniform distribution of
magnetic charge in space, in which non-associativity manifests as the breaking of angular
momentum conservation \cite{BakasLust}. This parallel between magnetic monopole physics and
non-geometric string backgrounds constitutes a compelling physical motivation for the
non-associative framework developed in this work. In a broader perspective, the role of
exceptional structures and non-associative algebras in quantum mechanics, supergravity, and
M/superstring theories has been reviewed comprehensively and recently by G\"unayd\i n
\cite{gunaydin2025}, tracing the appearance of octonions and exceptional Jordan algebras from
the original 1934 Jordan--von Neumann--Wigner paper through modern M-theory
compactifications; in particular, the uplift of the non-associative R-flux algebra of string
theory to M-theory involves octonionic structures \cite{gunaydin2025}.

\medskip

Similar non-commutative structures also arise in the non-perturbative limit of string theory
or relativistic quantum theories of gravity. A common expectation for the purpose of unifying
gravity with quantum mechanics involves replacing the classical geometry of spacetime by a
non-commutative quantum geometry. The main motivation of the present work is to formulate
the non-commutative nature of spatial coordinates in terms of the non-associativity of
observables, based on the algebraic formulation of quantum mechanics.

\medskip

In this paper, our main example will be the Landau problem that involves the motion of an
electron interacting with a constant, uniform background magnetic field. From the physical
point of view, the Landau problem is one of the most engaging problems in quantum mechanics,
inspired by ideas in diverse areas of physics such as quantum optics and condensed matter
physics. The minimal coupling of the electron charge density to an external magnetic field is
described in terms of the electromagnetic vector potential $\boldsymbol{\mathcal{A}}$ and
leads to the Hamiltonian in the symmetric gauge
$$
H=\frac{1}{2\mu}\!\left(\mathbf{p}-\frac{e}{c}\boldsymbol{\mathcal{A}}\right)^{\!2}
=:\frac{1}{2\mu}\mathbf{P}^2,
$$
where $\mathbf{P}=\mathbf{p}-\frac{e}{c}\boldsymbol{\mathcal{A}}$ denotes the kinetic
momenta. Thus, the Landau problem relates to the quantum mechanics of a charged particle on
the transverse $xy$-plane in the presence of a constant, uniform magnetic field present along
the $z$-axis. In metals, electrons occupy Landau energy levels given by
$$
E_N=\hbar\omega_L\!\left(N+\frac{1}{2}\right),
$$
where the Larmor frequency $\omega_L=\frac{eB}{\mu c}$, and $N=n+\frac{m-|m|}{2}$. The
infinitely degenerate Landau levels are labeled by two quantum numbers: the radial quantum
number $n\geq 0$ and the magnetic quantum number $m$ with range $m\geq -n$. The energy
spectrum is bounded from below: $N\geq 0$.

\medskip

Since the Landau problem finds numerous applications in different areas of physics, it has
been studied extensively in the literature. In particular, the first author et al.~\cite{DNP}
showed recently that a Newton--Hooke duality exists between the $2D$ hydrogen atom and the
Landau problem via the Tits--Kantor--Koecher construction of the conformal symmetries of the
Jordan algebra of real symmetric $2\times 2$ matrices. The Landau problem on the torus,
together with the non-commutative tori of magnetic translations and their Bloch wave
bi-module structure, was studied in~\cite{DP}, revealing the role of the metaplectic group
$Mp(2,\mathbb{R})$ in the quantum optics squeezing operators associated with these magnetic
translation algebras.

\medskip

Non-commuting spatial coordinates and fields can be seen in actual physical situations
\cite{jackiw2}. Landau models and their quantum Hall limit have become an instrument of
intense research activity as a physical realization of one of the simplest examples of
non-commutative geometries \cite{DJ}, \cite{DJT}, \cite{DH1}, \cite{DH2}, \cite{horvathy1},
\cite{horvathy2}, \cite{OlmoPly}, \cite{AlGomKamPly}, \cite{giri}, \cite{AlCorHorPly},
\cite{HA}, \cite{horvathy3}, \cite{AlcaPly}. Generalizations of canonical commutation
relations in the traditional approach lead to non-commutative geometry and appear in quantum
mechanics with source-free magnetic fields. Letting the coordinates of a charged particle in a magnetic field take values in a
non-commutative plane, with an independent parameter $\theta$, leads to the non-commutative
Landau problem.
In \cite{AlcaPly}, a correspondence is established between the dynamics of the two-vortex
system and the non-commutative Landau problem. The chiral decomposition in the
non-commutative Landau problem is advocated in \cite{AlGomKamPly}, \cite{AlCorHorPly} and
further generalized to non-vanishing electric fields in \cite{horvathy3}. But instead of
studying the non-commutative Landau problem in its full details, our main goal in this paper
is to present an infinite dimensional Jordan algebraic setting for quantum mechanics and then
study the Landau problem as a complete example in its full generality.

\medskip

Given an electron moving in a plane in the presence of a constant external magnetic field
$\mathbf{B}$ perpendicular to that plane, the spectrum of the quantized theory gives
infinitely degenerate Landau levels, with separation
$\mathcal{O}\!\left(\frac{eB}{\mu c}\right)$. The limit $B\rightarrow\infty$ or
$\mu\rightarrow 0$ projects onto the lowest Landau level. On each of the projected Landau
levels, one obtains a non-commuting algebra for the space coordinates giving the center of
the cyclotron motion \cite{LectureNotes}:
$$
[X,Y]=-\frac{i\hbar c}{eB}=:i\theta_0\neq 0.
$$
Motivated by this, we consider below the non-commutative Landau problem, in which the
planar coordinates themselves satisfy $[x,y]=i\theta$ with $\theta$ an independent parameter,
and we prove that $\theta$ can be cast as non-associativity of observables based on a corresponding
Jordan algebraic formulation.

\medskip

In what follows, we formulate the Landau problem coupled to a harmonic potential in the
Jordan algebraic framework developed below. It turns out that, in this problem, one obtains
a pair of commuting Hamiltonians, which can be written in terms of two pairs of mutually
commuting creation and annihilation operators of oscillator type. These operators further
generate two mutually commuting von Neumann algebras. Based on the Hilbert-Schmidt operator
formulation of the Landau problem, it is possible to write down Jordan algebraic formulas.
We take the associative case as our guideline and point out parallels to the associative case
whenever possible. It turns out that this problem at hand results in a composite system. The
Jordan algebraic approach has the particular advantage that, given the density matrix of the
composite system, one can identify the state vectors corresponding to the subsystems. We
accomplish this by defining a trace partially on the ambient space. Finally at the end, we
are able to derive a Schr\"{o}dinger-like time evolution equation that involves Jordan
algebraic associators for this specific problem, and present its solution explicitly. Our main aim is to give a non-associative, infinite dimensional, probabilistic (due to
the Hilbert space construction presented in Section~3) description of a concrete physical
model. We feel that our contribution
may potentially ignite further interesting applications of non-associative algebras through
this approach.

\medskip

The paper is organized as follows. In Section~2 we review the relevant aspects of operator
algebras, including $C^*$-algebras, von Neumann algebras, and their Jordan analogues
($JB$-algebras and $JBW$-algebras). The reader should keep in mind that the self-adjoint
part of a $C^*$-algebra forms a $JB$-algebra, and the self-adjoint part of a von Neumann
algebra forms a $JBW$-algebra, so many results in the Jordan context are generalizations of
the associative case. In Section~3 we formulate quantum mechanics using the Jordan algebraic
approach. After a review of the usual formalism of quantum mechanics and the postulates of
von Neumann \cite{V.Neumann}, we give the axiomatic approach of \cite{Emch} and present the
Jordan algebraic Hilbert space construction as an analogue of the Gelfand--Naimark--Segal
construction. In Section~4 the non-commutative Landau problem coupled to a harmonic potential
is discussed using the Jordan algebraic approach. It turns out that in this problem
one obtains a pair of commuting Hamiltonians which can be written in terms of two pairs of
mutually commuting creation and annihilation operators of oscillator type; these operators
further generate two mutually commuting von Neumann algebras. Based on the Hilbert-Schmidt
operator formulation of the non-commutative Landau problem, we write down the Jordan
algebraic formulation, taking the associative case as our guideline and mentioning connections
to the associative case whenever possible. In Section~5 density matrices and pure states are
characterized in the JBW-algebra setting. In Section~6 we derive the Jordan--Schr\"odinger
time evolution equation and present its solution explicitly. Section~7 contains the
conclusion.

\bigskip



\section{Operator Algebras}

\subsection{Basics on $C^*$-Algebras}

In this section we denote a separable complex Hilbert space of dimension $N$ (finite or
infinite) by $\mathscr{H}$, and we write $\mathcal{B}(\mathscr{H})$ for the $C^*$-algebra
of all bounded operators on $\mathscr{H}$.

\begin{definition}
Let $\mathcal{A}$ be an algebra. A mapping ${}^{*}:\mathcal{A}\rightarrow\mathcal{A}$,
defined by $A\mapsto A^{*}$ for every $A\in\mathcal{A}$, is called an
\emph{involution} or \emph{adjoint operation} of $\mathcal{A}$ if it satisfies the
following properties for all $A,B\in\mathcal{A}$ and $\alpha,\beta\in\mathbb{C}$:
\begin{itemize}
\item $A^{**}=A$,
\item $(AB)^{*}=B^{*}A^{*}$,
\item $(\alpha A+\beta B)^{*}=\bar{\alpha}A^{*}+\bar{\beta}B^{*}$.
\end{itemize}
\end{definition}

\begin{definition}
An algebra equipped with an involution is called a \emph{$*$-algebra}. A subset
$\widetilde{\mathcal{A}}$ of $\mathcal{A}$ is called \emph{self-adjoint} if
$A\in\widetilde{\mathcal{A}}$ implies $A^{*}\in\widetilde{\mathcal{A}}$.
\end{definition}

\begin{definition}
An algebra $\mathcal{A}$ is called a \emph{normed algebra} if for each element
$A\in\mathcal{A}$ there is associated a real number $\|A\|$, the norm of $A$, satisfying
for all $A,B\in\mathcal{A}$ and $\alpha\in\mathbb{C}$:
\begin{itemize}
\item $\|A\|\geq 0$, and $\|A\|=0$ if and only if $A=0$,
\item $\|\alpha A\|=|\alpha|\|A\|$,
\item $\|A+B\|\leq\|A\|+\|B\|$,
\item $\|AB\|\leq\|A\|\|B\|$.
\end{itemize}
\end{definition}

The norm defines a metric topology on $\mathcal{A}$ called the \emph{uniform topology}.
The $\varepsilon$-neighborhood of $A\in\mathcal{A}$ is
$$
N_{\varepsilon}(A)=\{B\in\mathcal{A}:\|B-A\|<\varepsilon\}.
$$
If $\mathcal{A}$ is complete with respect to the uniform topology it is called a \emph{Banach algebra}. A complete, normed algebra with adjoint operation satisfying
$\|A\|=\|A^{*}\|$ for all $A\in\mathcal{A}$ is called a \emph{Banach $*$-algebra}.

\begin{definition}
A \emph{$C^{*}$-algebra} $\mathcal{A}$ is a Banach $*$-algebra satisfying
$$
\|A^{*}A\|=\|A\|^{2}
$$
for all $A\in\mathcal{A}$.
\end{definition}

In the theory of quantum mechanics, the concepts of representations and states play a
central role.

\begin{definition}
Let $\mathcal{A}$ and $\widetilde{\mathcal{A}}$ be two $*$-algebras. A
\emph{$*$-morphism} between them is a map
$\pi:A\in\mathcal{A}\mapsto\pi(A)\in\widetilde{\mathcal{A}}$ satisfying for all
$A,B\in\mathcal{A}$ and $\alpha,\beta\in\mathbb{C}$:
\begin{itemize}
\item $\pi(\alpha A+\beta B)=\alpha\pi(A)+\beta\pi(B)$,
\item $\pi(AB)=\pi(A)\pi(B)$,
\item $\pi(A^{*})=\pi(A)^{*}$.
\end{itemize}
\end{definition}

It is worth remarking that every $*$-morphism is positive: if $A=B^{*}B$ is positive,
then $\pi(A)=\pi(B)^{*}\pi(B)\geq 0$.

\begin{definition}
A \emph{representation} of a $C^{*}$-algebra $\mathcal{A}$ is a pair $(\mathscr{H},\pi)$,
where $\mathscr{H}$ is a complex Hilbert space and $\pi$ is a $*$-morphism of $\mathcal{A}$
into $\mathcal{B}(\mathscr{H})$. The representation is called \emph{faithful} if and only
if each of the following equivalent conditions holds:
\begin{itemize}
\item $\ker\pi=\{0\}$,
\item $\pi$ is a $*$-isomorphism between $\mathcal{A}$ and $\pi(\mathcal{A})$,
\item $\|\pi(A)\|=\|A\|$ for all $A\in\mathcal{A}$,
\item $\pi(A)>0$ for all $A>0$.
\end{itemize}
\end{definition}

The equivalence of the above conditions follows from the proposition below.

\begin{proposition}
Let $\mathcal{A}$ be a $C^{*}$-algebra, $\widetilde{\mathcal{A}}$ a Banach $*$-algebra
with identity, and $\pi$ a $*$-morphism of $\widetilde{\mathcal{A}}$ into $\mathcal{A}$.
Then $\pi$ is continuous and $\|\pi(A)\|\leq\|A\|$ for all $A\in\widetilde{\mathcal{A}}$.
Moreover, if $\widetilde{\mathcal{A}}$ is itself a $C^{*}$-algebra, then the range of
$\pi$ is a $C^{*}$-subalgebra of $\mathcal{A}$.
\end{proposition}

\begin{definition}
Let $A$ be an element of a Banach algebra $\mathcal{A}$ with identity. The
\emph{spectral radius} $\rho(A)$ of $A$ is
$$
\rho(A)=\sup\{|\lambda|:\lambda\in\sigma_{\mathcal{A}}(A)\},
$$
where the \emph{spectrum} $\sigma_{\mathcal{A}}(A)$ is the set of all $\lambda\in\mathbb{C}$
such that $(A-\lambda\mathbb{I})$ is not invertible.
\end{definition}

\begin{remark}
One has
$$
\rho(A)=\lim_{n\to\infty}\|A^{n}\|^{1/n}=\inf_{n}\|A^{n}\|^{1/n}\leq\|A\|,
$$
and the spectrum of $A$ is a nonempty compact subset of $\mathbb{C}$.
\end{remark}

\begin{definition}
A \emph{cyclic representation} of a $C^{*}$-algebra $\mathcal{A}$ is a triplet
$(\mathscr{H},\pi,\Omega)$, where $(\mathscr{H},\pi)$ is a representation of $\mathcal{A}$
and $\Omega$ is a \emph{cyclic vector} for $\pi$, meaning that
$\{\pi(A)\Omega:A\in\mathcal{A}\}$ is dense in $\mathscr{H}$. A closed subspace
$\mathscr{R}\subseteq\mathscr{H}$ is called a \emph{cyclic subspace} if
$\bigl\{\sum_i\pi(A_i)\psi_i:A_i\in\mathcal{A},\,\psi_i\in\mathscr{R}\bigr\}$ is dense
in $\mathscr{H}$.
\end{definition}

\begin{definition}
A linear functional $\omega$ over the $C^{*}$-algebra $\mathcal{A}$ is \emph{positive} if
$\omega(A^{*}A)\geq 0$ for all $A\in\mathcal{A}$. A positive linear functional of unit
norm is called a \emph{state}.
\end{definition}

Note that every positive element of a $C^{*}$-algebra is of the form $A^{*}A$, so
positivity of a state $\omega$ is equivalent to $\omega$ being non-negative on all
positive elements of $\mathcal{A}$.

\begin{definition}
Given a representation $(\mathscr{H},\pi)$ of $\mathcal{A}$ and a nonzero vector
$\Omega\in\mathscr{H}$, define $\omega_{\Omega}$ by
$$
\omega_{\Omega}(A)=\langle\Omega,\pi(A)\Omega\rangle
$$
for all $A\in\mathcal{A}$. Then $\omega_{\Omega}$ is a positive linear functional since
$\omega_{\Omega}(A^{*}A)=\|\pi(A)\Omega\|^{2}\geq 0$. When $\|\Omega\|=1$ and $\pi$
is non-degenerate (i.e.\ $\{\pi(A)h:A\in\mathcal{A},h\in\mathscr{H}\}$ is dense in
$\mathscr{H}$), $\omega_{\Omega}$ is called a \emph{vector state} for the representation.
The cyclic representation $(\mathscr{H}_{\omega},\pi_{\omega},\Omega_{\omega})$
constructed from a state $\omega$ is called the \emph{canonical cyclic representation}
of $\mathcal{A}$ associated with $\omega$.
\end{definition}

\begin{definition}
A state $\omega$ over a $C^{*}$-algebra is \emph{pure} if, for any positive linear
functional $\psi:\mathcal{A}\to\mathbb{C}$ with $\psi\leq\omega$, there exists some $\lambda\in[0,1]$ such that $\psi=\lambda\omega$.
\end{definition}

\begin{theorem}
Let $\omega$ be a state over the $C^{*}$-algebra $\mathcal{A}$ and
$(\mathscr{H}_{\omega},\pi_{\omega},\Omega_{\omega})$ the associated canonical cyclic
representation. Then the following conditions are equivalent:
\begin{itemize}
\item $(\mathscr{H}_{\omega},\pi_{\omega})$ is irreducible,
\item $\omega$ is pure,
\item $\omega$ is an extreme point of the set of states over $\mathcal{A}$.
\end{itemize}
\end{theorem}

\subsection{Von Neumann Algebras}

Let $\mathscr{H}$ be a Hilbert space. For each subset $\mathcal{A}$ of
$\mathcal{B}(\mathscr{H})$, the \emph{commutant} $\mathcal{A}'$ is the set of all
bounded operators on $\mathscr{H}$ that commute with every element of $\mathcal{A}$.

\begin{definition}
A \emph{von Neumann algebra} is a $*$-subalgebra $\mathcal{A}$ of $\mathcal{B}(\mathscr{H})$
satisfying $\mathcal{A}=\mathcal{A}''$ (the bicommutant condition). Equivalently,
$\mathcal{A}$ is a $C^{*}$-algebra acting on $\mathscr{H}$ that is closed in the
\emph{weak operator topology}:
$$
\lim_{n\to\infty}A_{n}=A
\quad\Longleftrightarrow\quad
\lim_{n\to\infty}\langle\xi|A_{n}\eta\rangle=\langle\xi|A\eta\rangle
\quad\text{for all }\xi,\eta\in\mathscr{H}.
$$
Equivalently, under the $\sigma$-weak topology,
$\lim_{n\to\infty}A_{n}=A$ if and only if for all sequences $(\xi_{k}),(\eta_{k})$
in $\mathscr{H}$ with $\sum_{k}\|\xi_{k}\|^{2}<\infty$ and
$\sum_{k}\|\eta_{k}\|^{2}<\infty$ we have
$$
\lim_{n\to\infty}\sum_{k=1}^{\infty}\langle\xi_{k}|A_{n}\eta_{k}\rangle
=\sum_{k=1}^{\infty}\langle\xi_{k}|A\eta_{k}\rangle.
$$
\end{definition}

\begin{remark}
A von Neumann algebra $\mathcal{A}$ always contains the identity $\mathbb{I}_{\mathscr{H}}$.
It is called a \emph{factor} if its center satisfies
$\mathcal{A}\cap\mathcal{A}'=\mathbb{C}\,\mathbb{I}_{\mathscr{H}}$.
\end{remark}

The algebra $\mathcal{B}(\mathscr{H})$ is itself a von Neumann algebra and a factor.
The $C^{*}$-norm property follows from
$$
\|A\|^{2}
=\sup_{x\neq 0}\frac{\|Ax\|^{2}_{\mathscr{H}}}{\|x\|^{2}_{\mathscr{H}}}
=\sup_{x\neq 0}\frac{\langle A^{*}Ax,x\rangle}{\|x\|^{2}_{\mathscr{H}}}
\leq\sup_{x\neq 0}\frac{\|A^{*}Ax\|_{\mathscr{H}}}{\|x\|_{\mathscr{H}}}
=\|A^{*}\|\|A\|
$$
for all $A\in\mathcal{B}(\mathscr{H})$.

\begin{definition}
Let $\mathcal{A}$ be a von Neumann algebra. A bounded linear functional
$\phi:\mathcal{A}\to\mathbb{C}$, denoted $\langle\phi;A\rangle$, is called a
\emph{state} on $\mathcal{A}$ if:
\begin{itemize}
\item $\langle\phi;A^{*}A\rangle\geq 0$ for all $A\in\mathcal{A}$,
\item $\langle\phi;\mathbb{I}_{\mathscr{H}}\rangle=1$.
\end{itemize}
The state $\phi$ is called a \emph{vector state} if there exists $\Phi\in\mathscr{H}$
with $\langle\phi;A\rangle=\langle\Phi|A\Phi\rangle$ for all $A\in\mathcal{A}$;
such a state is automatically \emph{normal}. A state $\omega$ is \emph{faithful} if
$\omega(A)>0$ whenever $A>0$.
\end{definition}

\begin{remark}
Let $\mathscr{H}$ be a separable Hilbert space. Every normal state $\omega$ over
$\mathcal{B}(\mathscr{H})$ is of the form $\omega(A)=\mathrm{tr}(\rho A)$, where $\rho$
is a density matrix (positive, trace-class, $\mathrm{tr}(\rho)=1$). Moreover, $\omega$ is faithful if and only if $\ker\rho=\{0\}$ (in infinite dimensions a trace-class $\rho$ is never boundedly invertible).
\end{remark}

\begin{proposition}
Let $\mathrm{tr}$ be the usual trace on $\mathcal{B}(\mathscr{H})$ and
$\mathcal{T}(\mathscr{H})$ the Banach space of trace-class operators with the trace norm.
Then $\mathcal{B}(\mathscr{H})$ is the Banach-space dual of $\mathcal{T}(\mathscr{H})$
via the duality
$$
(A,T)\in\mathcal{B}(\mathscr{H})\times\mathcal{T}(\mathscr{H})\;\longmapsto\;\mathrm{tr}(AT).
$$
\end{proposition}

We close this subsection with a theorem that will be central in the subsequent quantum
mechanical discussion.

\begin{theorem}
\label{thm:normal-states}
Let $\omega$ be a state on a von Neumann algebra $\mathcal{A}$ acting on $\mathscr{H}$.
The following conditions are equivalent:
\begin{enumerate}
\item[\rm(i)] $\omega$ is normal.
\item[\rm(ii)] $\omega$ is $\sigma$-weakly continuous, i.e.\ there exist sequences
$\{\xi_{n}\},\{\eta_{n}\}$ in $\mathscr{H}$ with $\sum_{n}\|\xi_{n}\|^{2}<\infty$ and
$\sum_{n}\|\eta_{n}\|^{2}<\infty$ such that
$\omega(A)=\sum_{n=1}^{\infty}\langle\xi_{n}|A\eta_{n}\rangle$.
\item[\rm(iii)] There exists a density matrix $\rho$ (a positive, trace-class operator
with $\mathrm{tr}(\rho)=1$) such that $\omega(A)=\mathrm{tr}(\rho A)$ for all
$A\in\mathcal{A}$.
\end{enumerate}
\end{theorem}

\subsection{Jordan-Banach Algebras}

The physically motivating example of a $JB$-algebra is a $JC$-algebra: the norm-closed,
real, linear space of self-adjoint operators on a Hilbert space $\mathscr{H}$, closed
under the \emph{Jordan product}
\begin{equation}
\label{eq:jordan-product}
A\bullet B=\frac{1}{2}(AB+BA).
\end{equation}
This product is bilinear, commutative, and non-associative.

\begin{definition}
A \emph{Jordan algebra} over $\mathbb{R}$ is a real vector space $\mathcal{A}$ with a
commutative, bilinear product $\bullet$ satisfying the \emph{Jordan identity}
\begin{equation}
\label{eq:jordan-identity}
(A^{2}\bullet B)\bullet A=A^{2}\bullet(B\bullet A)
\end{equation}
for all $A,B\in\mathcal{A}$.
\end{definition}

If $\mathcal{A}$ is any associative algebra over $\mathbb{R}$, or any subspace of one
closed under \eqref{eq:jordan-product}, then it becomes a Jordan algebra under this
product. Jordan algebras isomorphic to such algebras are called \emph{special}; those
that cannot be embedded into any associative algebra are called \emph{exceptional}. If
$\mathcal{A}$ is a $C^{*}$-algebra, its self-adjoint part $\mathcal{A}_{sa}$ is a
special Jordan algebra under \eqref{eq:jordan-product}.

In any Jordan algebra, powers are defined recursively by $A^{1}=A$ and
$A^{n}=A^{n-1}\bullet A$. The following theorem, due to Jordan, von Neumann and Wigner
\cite{JNW}, characterizes power-associativity.

\begin{theorem}[\cite{JNW}]
For a formally real Jordan algebra $\mathcal{A}$ (i.e.\ one satisfying
$A^{2}+B^{2}+\cdots=0\Rightarrow A=B=\cdots=0$), the following are equivalent:
\begin{enumerate}
\item[\rm(i)] $A^{n}\bullet A^{m}=A^{n+m}$ for all $n,m\geq 1$,
\item[\rm(ii)] $A^{2}\bullet(B\bullet A)=(A^{2}\bullet B)\bullet A$ for all
$A,B\in\mathcal{A}$.
\end{enumerate}
\end{theorem}

The classification of finite dimensional, formally real Jordan algebras is also
established in \cite{JNW}.

\begin{theorem}[\cite{JNW}]
Every finite dimensional, formally real Jordan algebra is a direct sum of simple ideals.
There are five types of simple algebras: four types of Hermitian matrix algebras
corresponding to the four real division algebras $\mathbb{R},\mathbb{C},\mathbb{H},
\mathbb{O}$, together with the spin factors. More precisely, every such simple algebra
is isomorphic to one of
$$
\mathcal{H}_{n}(\mathbb{R}),\quad
\mathcal{H}_{n}(\mathbb{C}),\quad
\mathcal{H}_{n}(\mathbb{H}),\quad
\mathcal{H}_{3}(\mathbb{O}),\quad
JSpin_{n},
$$
where $\mathcal{H}_{n}(\mathbb{K})$ denotes $n\times n$ Hermitian matrices over $\mathbb{K}$.
\end{theorem}

The algebra $\mathcal{H}_{3}(\mathbb{O})$, consisting of $3\times 3$ Hermitian matrices
with octonionic entries, is called the \emph{Albert algebra}. Albert \cite{Albert} showed
that it is an exceptional Jordan algebra of dimension $27$. Zel'manov \cite{Zelmanov}
later gave a complete treatment, showing that finite dimensional Albert algebras are the
only simple exceptional Jordan algebras.

Since the algebras arising in quantum mechanics are typically infinite dimensional, one
needs a structure theory for infinite dimensional Jordan algebras of self-adjoint
operators. This is achieved by Alfsen, Shultz and Størmer \cite{ASS} through the
introduction of Jordan-Banach algebras.

\begin{definition}
A \emph{Jordan-Banach algebra} ($JB$-algebra) is a Jordan algebra $\mathcal{A}$ over
$\mathbb{R}$ with identity $\mathbb{I}$, equipped with a complete norm satisfying for
all $A,B\in\mathcal{A}$:
\begin{align}
&\|A\bullet B\|\leq\|A\|\|B\|, \label{eq:JB1}\\
&\|A^{2}\|=\|A\|^{2}, \label{eq:JB2}\\
&\|A^{2}\|\leq\|A^{2}+B^{2}\|. \label{eq:JB3}
\end{align}
\end{definition}

The self-adjoint part of any $C^{*}$-algebra is a $JB$-algebra under
\eqref{eq:jordan-product}, with \eqref{eq:JB1}--\eqref{eq:JB3} following from the
$C^{*}$-algebra axioms. Furthermore, any norm-closed Jordan subalgebra of
$\mathcal{B}(\mathscr{H})_{sa}$ is a $JB$-algebra.

\begin{definition}
A \emph{$JC$-algebra} is a $JB$-algebra isomorphic to a norm-closed Jordan subalgebra
of $\mathcal{B}(\mathscr{H})_{sa}$. A finite dimensional $JB$-algebra is a $JC$-algebra
if and only if it contains no Albert algebra as a direct summand. In infinite dimensions,
every $JB$-algebra $\mathcal{A}$ admits an exceptional ideal $\mathcal{J}$ such that
$\mathcal{A}/\mathcal{J}$ is a $JC$-algebra.
\end{definition}

\begin{theorem}[Gelfand--Naimark for Jordan algebras \cite{ASS}]
Every $JB$-algebra can be isometrically and isomorphically embedded in
$$
\mathcal{B}(\mathscr{H})_{sa}\oplus\mathcal{C}(X,\mathcal{H}_{3}(\mathbb{O})),
$$
a direct sum of the self-adjoint part of the bounded operators on some real Hilbert space
and the purely exceptional algebra of continuous Albert-algebra-valued functions on some
compact space $X$.
\end{theorem}

Every $JB$-algebra is not only a Banach space but also an ordered linear space, with the
positive cone being the set of squares.

\begin{definition}
An \emph{order unit space} is an ordered normed linear space $\mathcal{A}$ over
$\mathbb{R}$ with a closed positive cone and a distinguished element $e$ (the
\emph{order unit}) satisfying
$$
\|a\|=\inf\{\lambda>0:-\lambda e\leq a\leq\lambda e\}.
$$
\end{definition}

With the ordering above, every $JB$-algebra becomes an order unit space. The Jordan
identity \eqref{eq:jordan-identity} was originally proposed as the fundamental axiom of
the model, but it has no direct physical interpretation. Von Neumann \cite{Neumann} and
later Iochum and Loupias \cite{IL1,IL2} proposed replacing it by power-associativity,
leading to the notion of a \emph{Banach power-associative algebra} (also called an
\emph{order unit algebra}).

\begin{definition}
An \emph{order unit algebra} is an order unit space $\mathcal{A}$ that is also a
power-associative, complete, normed algebra (not necessarily associative) in which the
distinguished order unit $\mathbb{I}$ is a multiplicative identity and
$A^{2}\in\mathcal{A}^{+}$ for all $A\in\mathcal{A}$.
\end{definition}

\begin{theorem}
Every $JB$-algebra $\mathcal{A}$ is a commutative order unit algebra with the given norm
and positive cone $\mathcal{A}^{+}=\{A^{2}:A\in\mathcal{A}\}$. Conversely, every
commutative order unit algebra is a $JB$-algebra.
\end{theorem}

For each $A\in\mathcal{A}$, the norm-closed subalgebra $C(A,\mathbb{I})$ generated by
$A$ and the identity is associative, and the following proposition is vital in the spectral
theory of $JB$-algebras.

\begin{proposition}
Let $\mathcal{A}$ be a $JB$-algebra and $\mathcal{B}=C(A,\mathbb{I})$ for some
$A\in\mathcal{A}$. Then $\mathcal{B}$ is isometrically isomorphic to $C_{\mathbb{R}}(X)$,
the algebra of real continuous functions on some compact Hausdorff space $X$.
\end{proposition}

In any Jordan algebra, the \emph{triple product} is defined by
\begin{equation}
\label{eq:triple-product}
\{ABC\}=(A\bullet B)\bullet C+(B\bullet C)\bullet A-(A\bullet C)\bullet B.
\end{equation}
The triple product is linear in each factor and satisfies $\{ABC\}=\{CBA\}$. In a
special Jordan algebra, $\{ABC\}=\frac{1}{2}(ABC+CBA)$ (using the ambient associative
product). The \emph{$U$-operator} is defined by
$$
U_{A}(B)=\{ABA\}=2A\bullet(A\bullet B)-A^{2}\bullet B,
$$
which for a special Jordan algebra reduces to $U_{A}(B)=ABA$. The following
\emph{fundamental identities} hold in any Jordan algebra:
\begin{align}
\{\{ABA\}C\{ABA\}\}&=\{A\{B\{ACA\}B\}A\}, \label{eq:fund1}\\
\{BAB\}^{2}&=\{B\{AB^{2}A\}B\}, \label{eq:fund2}
\end{align}
which can be rewritten as $U_{\{ABA\}}=U_{A}U_{B}U_{A}$.

The following theorem shows that most Jordan identities for special algebras hold in full
generality.

\begin{theorem}[MacDonald's theorem]
Any polynomial Jordan identity in three or fewer variables, of degree at most $1$ in one
of them, that holds for all special Jordan algebras holds for all Jordan algebras.
\end{theorem}

\begin{definition}
An element $A$ of a $JB$-algebra $\mathcal{A}$ is \emph{invertible} if it is invertible
in the Banach algebra $C(A,\mathbb{I})$; its inverse is denoted $A^{-1}$.
\end{definition}

\begin{proposition}
Let $\mathcal{A}$ be an associative algebra with associated Jordan algebra under
\eqref{eq:jordan-product}. For $A,B\in\mathcal{A}$, $AB=BA=\mathbb{I}$ if and only if
$$
A\bullet B=\mathbb{I}\quad\text{and}\quad A^{2}\bullet B=A.
$$
\end{proposition}

If $A,B$ satisfy the above, they are called \emph{Jordan invertible}, and $B$ is the
Jordan inverse of $A$. The following proposition shows that the two notions coincide.

\begin{proposition}
For elements $A,B$ in a Jordan algebra $\mathcal{A}$, the following are equivalent:
\begin{itemize}
\item $B$ is the inverse of $A$ in the Banach algebra $C(A,\mathbb{I})$,
\item $B$ is the Jordan inverse of $A$ in $\mathcal{A}$.
\end{itemize}
\end{proposition}

\begin{definition}
The \emph{spectrum} of $A\in\mathcal{A}$ is the set
$\sigma(A)=\{\lambda\in\mathbb{R}:\lambda\mathbb{I}-A\text{ is not invertible}\}$.
\end{definition}

\begin{definition}
Positive elements $A,B\in\mathcal{A}$ are \emph{orthogonal} if
$\{ABA\}=\{BAB\}=0$.
\end{definition}

\subsection{JBW-Algebras}

$JBW$-algebras are the Jordan analogues of von Neumann ($W^{*}$) algebras; they include
the self-adjoint part of every von Neumann algebra as a special case.

\begin{definition}
An ordered Banach space $\mathcal{A}$ is \emph{monotone complete} if every increasing
net $\{A_{\alpha}\}$ bounded above has a least upper bound $A\in\mathcal{A}$. A bounded
linear functional $\sigma$ on a monotone complete space is \emph{normal} if
$\sigma(A_{\alpha})\to\sigma(A)$ for every such net.
\end{definition}

\begin{definition}
A \emph{$JBW$-algebra} is a $JB$-algebra that is monotone complete and admits a
separating set of normal states.
\end{definition}

This definition is modelled on Kadison's characterization of $W^{*}$-algebras as those
$C^{*}$-algebras whose self-adjoint parts are monotone complete and admit a separating
set of normal states. The self-adjoint part of any $W^{*}$-algebra is a $JBW$-algebra.

\begin{definition}
Let $\mathcal{A}$ be a $JBW$-algebra with normal state space $\mathcal{K}$ and linear
span $\mathcal{V}=\mathrm{span}\,\mathcal{K}$. The \emph{$\sigma$-weak topology} on
$\mathcal{A}$ is defined by the duality between $\mathcal{A}$ and $\mathcal{V}$, with
seminorms $A\mapsto|\rho(A)|$ for $\rho\in\mathcal{K}$. The \emph{$\sigma$-strong
topology} is defined by seminorms $A\mapsto\rho(A^{2})^{1/2}$; equivalently,
$A_{\alpha}\to A$ $\sigma$-strongly if $\rho((A_{\alpha}-A)^{2})\to 0$ for all
$\rho\in\mathcal{K}$.
\end{definition}

The $\sigma$-weak and $\sigma$-strong topologies coincide with the corresponding
topologies in the von Neumann algebra setting.

\begin{lemma}[\cite{AS,ASSS}]
\label{lem:density-projections}
Let $\mathcal{A}$ be a commutative, order unit algebra that is the dual of a base norm
space $V$ with separately $w^{*}$-continuous multiplication. Then for each $a\in\mathcal{A}$
and $\varepsilon>0$ there exist orthogonal projections $p_{1},p_{2},\ldots,p_{n}$ in the
$w^{*}$-closed subalgebra generated by $a$ and $\mathbb{I}$, and scalars
$\lambda_{1},\ldots,\lambda_{n}\in\mathbb{R}$, such that
$$
\Bigl\|a-\sum_{i=1}^{n}\lambda_{i}p_{i}\Bigr\|<\varepsilon.
$$
\end{lemma}

\begin{theorem}[\cite{ASSS}]
\label{thm:bidual-JBW}
Let $\mathcal{A}$ be a commutative order unit algebra that is a $JB$-algebra. Then the
set of finite linear combinations of orthogonal projections is norm dense in
$\mathcal{A}^{**}$, where $\mathcal{A}^{**}$ denotes the bidual of $\mathcal{A}$; this
bidual is a $JBW$-algebra.
\end{theorem}

On a semi-finite $JBW$-algebra $\mathcal{A}$ one can work with normal, faithful, semi-finite traces
\cite{Iochum}. The trace on a Jordan algebra is defined as a weight on $\mathcal{A}$
satisfying the \emph{3-cyclicity condition}
\begin{equation}
\label{eq:trace-cyclicity}
\mathrm{tr}\bigl(A\bullet(B\bullet C)\bigr)=\mathrm{tr}\bigl((A\bullet B)\bullet C\bigr),
\end{equation}
together with linearity and positivity:
$$
\mathrm{tr}(A+B)=\mathrm{tr}(A)+\mathrm{tr}(B),\qquad
\mathrm{tr}(\lambda A)=\lambda\,\mathrm{tr}(A),\qquad\lambda>0.
$$
The 3-cyclicity condition \eqref{eq:trace-cyclicity} is the Jordan analogue of the
tracial cyclicity $\mathrm{tr}(AB)=\mathrm{tr}(BA)$ in the associative setting, and it
plays a fundamental role in the Hilbert space construction presented in the next section.


\section{Jordan Algebraic Formulation of Quantum Mechanics}

\subsection{Standard Formulation of Quantum Mechanics}

Dirac's \cite{Dirac} exposition of the fundamentals of quantum mechanics forms the
traditional framework, while von Neumann \cite{V.Neumann} reformulated these ideas with
a different mathematical approach. We briefly review this framework before developing the
Jordan algebraic approach.

Given a physical system, observables are represented by self-adjoint, linear operators
acting on a Hilbert space $\mathscr{H}$, a finite or infinite dimensional complex vector
space equipped with a scalar product $\langle\Psi,\Phi\rangle$ (antilinear in the first
component, linear in the second). $\mathscr{H}$ is complete with respect to the induced
norm $\|\Psi\| := \langle\Psi,\Psi\rangle^{1/2}$.

States of a physical system are identified with positive, self-adjoint operators $\rho$
of unit trace acting on $\mathscr{H}$, called \emph{density matrices}. The expectation
value of an observable $A$ in state $\rho$ is
$$
\langle A\rangle_{\rho} = \mathrm{tr}(\rho A).
$$
For every unit vector $\Psi\in\mathscr{H}$, there corresponds a \emph{pure state} $\psi$
defined by
$$
\langle A\rangle_{\psi} = \mathrm{tr}(P_{\psi}A) = \langle\Psi,A\Psi\rangle,
$$
where $P_{\psi} = |\Psi\rangle\langle\Psi|$ is the orthogonal projector onto $\Psi$.
The \emph{transition probability} from state $\phi$ to state $\psi$ is
$$
\Pr\{\psi|\phi\} = \langle\Phi, P_{\psi}\Phi\rangle = |\langle\Phi,\Psi\rangle|^2.
$$
When $\mathscr{H} = L^2(\mathbb{R}^n)$, $\Psi$ is interpreted as the wave function of
the state, determined up to a phase. A general density matrix is a statistical mixture
$$
\rho = \sum_{n=0}^{\infty} p_n P_{\psi_n}, \quad p_n\in[0,1],\quad \sum_{n=0}^{\infty}p_n = 1.
$$

In the Schr\"{o}dinger picture, time evolution is described by a strongly continuous
one-parameter group of unitary operators $\{U_t\}_{t\in\mathbb{R}}$:
$$
\Psi_t = U_t\Psi, \qquad \rho_t = U_t\rho U_{-t}.
$$
(In the Heisenberg picture, observables evolve as $A_t = U_{-t}AU_t$, while states
remain fixed.) The expectation value at time $t$ is
$$
\langle A\rangle_{\psi}(t) = \langle\Psi_t, A\Psi_t\rangle, \qquad
\langle A\rangle_{\rho}(t) = \mathrm{tr}(\rho_t A).
$$
The differential forms of the evolution laws are
\begin{equation}
\label{eq:schrodinger}
\frac{d}{dt}\Psi_t = -iH\Psi_t,
\end{equation}
\begin{equation}
\label{eq:liouville}
\frac{d}{dt}\rho_t = -i[H,\rho_t] =: -i\mathcal{L}\rho_t,
\end{equation}
where $H$ is the Hamiltonian and $\mathcal{L}\rho = [H,\rho]$ is the Liouville operator.

\subsection{Jordan Algebra of Observables}

The algebraic approach to quantum mechanics traces back to the matrix formulation
associated with Dirac, Heisenberg, Born, and Jordan. Von Neumann \cite{V.Neumann} made
two key contributions: he formulated the theory as an eigenvalue problem and interpreted
states via probability theory. Here we present the axiomatic approach due to
\cite{Emch}, which leads naturally to the Jordan algebraic structure.

\subsubsection{Postulates and Structural Axioms}

\begin{postulate}
For each observable of a physical system, there corresponds a self-adjoint, linear
operator $A$ acting on a Hilbert space $\mathscr{H}$.
\end{postulate}

Denoting the set of observables by $\mathcal{A}$, one observes that: (i) real linear
combinations $\sum_i\lambda_iA_i$ belong to $\mathcal{A}$; (ii) for any $A\in\mathcal{A}$
and real function $f$, $f(A)\in\mathcal{A}$; (iii) while $AB$ need not belong to
$\mathcal{A}$ for non-compatible $A,B$, the symmetric product
$$
\{A,B\} = \frac{1}{2}(AB+BA)
$$
always belongs to $\mathcal{A}$ \cite{Emch}. This product is commutative and bilinear,
and can be written purely in terms of (i) and (ii) as
$$
\{A,B\} = \frac{1}{4}\bigl((A+B)^2 - (A-B)^2\bigr),
$$
so it is well defined within $\mathcal{A}$ without reference to $AB$.

\begin{postulate}
Each state $\phi$ of a physical system corresponds to a density matrix $\rho_\phi$ on
$\mathscr{H}$, and the expectation value is $\langle A\rangle_\phi = \mathrm{tr}(\rho_\phi A)$.
\end{postulate}

This postulate implies: (i) $\mathrm{tr}(\rho_\phi\sum_i\lambda_i A_i)
= \sum_i\lambda_i\langle A_i\rangle_\phi$; (ii) $\mathrm{tr}(\rho_\phi A^2)\geq 0$;
(iii) $\mathrm{tr}(\rho_\phi\mathbb{I})=1$. 

\begin{remark}
For any convex combination of states, $\langle A\rangle_\phi := \sum_i\lambda_i
\langle A\rangle_{\phi_i}$ with $\lambda_i\geq 0$, $\sum_i\lambda_i=1$, defines a
state. States that cannot be decomposed in this way are called \emph{pure states}.
\end{remark}

We denote by $\mathcal{S}$ the set of all states satisfying properties (i)--(iii).
Postulates 1 and 2 are too restrictive for a purely algebraic formulation; the following
axioms replace them.

\begin{axiom}
For a given physical system, one associates the triple $(\mathcal{A},\mathcal{S},
\langle\cdot\rangle)$, where $\mathcal{A}$ is the set of all observables, $\mathcal{S}$
is the set of all states, and the map $\langle\cdot\rangle:(\mathcal{A},\mathcal{S})
\to\mathbb{R}$ gives the expectation value of each observable in each state.
\end{axiom}

States $\phi\in\mathcal{S}$ are therefore real-valued maps on $\mathcal{A}$, with
$\phi(A):=\langle A\rangle_\phi$. A state $\phi$ is \emph{dispersion-free} on $A$ if
all measurements of $A$ concentrate on $\langle A\rangle_\phi$; the set of such states
is $\mathcal{S}_A$, and $\sigma_A := \phi(A)|_{\phi\in\mathcal{S}_A}$ is the spectrum
of $A$.

\begin{definition}
A subset $\mathcal{I}\subseteq\mathcal{S}$ is \emph{complete} if it is full with respect
to $\mathcal{A}_\mathcal{I} := \{A\in\mathcal{A}:\mathcal{I}\subseteq\mathcal{S}_A\}$.
It is \emph{deterministic} for $\mathcal{B}\subseteq\mathcal{A}$ if
$\mathcal{B}\subseteq\mathcal{A}_\mathcal{I}$. A subset $\mathcal{B}\subseteq\mathcal{A}$
is \emph{compatible} if $\mathcal{S}_\mathcal{B} :=
\bigcap_{B\in\mathcal{B}}\mathcal{S}_B$ is complete.
\end{definition}

\begin{axiom}
The relation $\leq$ (defined by $A\leq B$ iff $\langle A\rangle_\phi\leq
\langle B\rangle_\phi$ for all $\phi\in\mathcal{S}$) is a partial order on $\mathcal{A}$.
\end{axiom}

\begin{axiom}
There exist $O,\mathbb{I}\in\mathcal{A}$ with $\langle O\rangle_\phi=0$ and
$\langle\mathbb{I}\rangle_\phi=1$ for all $\phi\in\mathcal{S}$; $\mathcal{A}$ is
closed under real scalar multiplication; and for any $A,B\in\mathcal{A}$ there exists
$A+B\in\mathcal{A}$ with $\langle A+B\rangle_\phi = \langle A\rangle_\phi +
\langle B\rangle_\phi$.
\end{axiom}

This makes $\mathcal{A}$ a real vector space and $\mathcal{S}$ the set of real linear
functionals on $\mathcal{A}$.

\begin{axiom}
\label{ax:deterministic}
For any $A\in\mathcal{A}$, the set $\mathcal{S}_A$ is deterministic for the subspace
generated by $A$; for any $A,B\in\mathcal{A}$: $\mathcal{S}_A\cap\mathcal{S}_B
\subseteq\mathcal{S}_{A+B}$; and $\mathcal{S}_\mathbb{I}=\mathcal{S}$.
\end{axiom}

This axiom implies $\sigma_A$ and $\mathcal{S}_A$ are nonempty, and allows the formulation
of powers.

\begin{axiom}
\label{ax:powers}
For any $A\in\mathcal{A}$ and non-negative integer $n$, there exists $A^n\in\mathcal{A}$
with $\mathcal{S}_A\subseteq\mathcal{S}_{A^n}$ and $\langle A^n\rangle_\phi =
\langle A\rangle_\phi^n$.
\end{axiom}

From these axioms one deduces: (i) $A^n$ is unique; (ii) power-raising is associative
and distributive over scalar multiplication; (iii) $A^2+B^2+\cdots=0$ implies
$A=B=\cdots=0$ (formal reality); (iv) $\mathscr{B}(A)$, the subalgebra generated by
powers of $A$, is a compatible set.

\begin{theorem}[\cite{Emch}]
Let $\mathcal{B}$ be a compatible set and $\mathscr{B}(\mathcal{B})$ the subalgebra
generated from elements of $\mathcal{B}$ by addition, power-raising, and scalar
multiplication. Then $\mathscr{B}(\mathcal{B})$ is a compatible set of observables.
\end{theorem}

\subsubsection{The Symmetric Product and the Jordan Structure}

With each pair $A,B\in\mathcal{A}$, associate
$$
A\bullet B := \frac{1}{2}\bigl((A+B)^2 - A^2 - B^2\bigr).
$$
This symmetric product is commutative. Using Axioms~\ref{ax:deterministic} and~\ref{ax:powers},
$$
\langle A\bullet B\rangle_\phi = \langle A\rangle_\phi\langle B\rangle_\phi
$$
for all $\phi\in\mathcal{S}_A\cap\mathcal{S}_B$. In particular, $A^n\bullet A^m =
A^{n+m}$ and $A\bullet\mathbb{I}=A$.

The \emph{associator} is $[A,B,C]:=(A\bullet B)\bullet C - A\bullet(B\bullet C)$.
For a compatible set $\mathcal{B}$, $[A,B,C]=0$ for all
$A,B,C\in\mathcal{A}(\mathcal{B})$. In a special Jordan algebra one has the relation
\begin{equation}
\label{eq:assoc-double-comm}
[A,B,C] = \frac{1}{4}[B,[A,C]],
\end{equation}
so that $A$ and $C$ being compatible (commuting) implies the associator vanishes.

\begin{theorem}[\cite{Emch}]
\label{thm:assoc-distributive}
The symmetric product is associative and distributive over $\mathcal{A}(\mathcal{B})$
for any compatible set $\mathcal{B}$.
\end{theorem}

\begin{axiom}
\label{ax:compatibility}
For any triple $A,B,C\in\mathcal{A}$ in which $A$ and $C$ are compatible, the associator
$[A,B,C]$ vanishes.
\end{axiom}

Since all powers of an observable are compatible, $[A^n,B,A^m]=0$ for all $n,m$, giving
in particular $A\bullet(B\bullet A^2)=(A\bullet B)\bullet A^2$, which is the Jordan
identity.

\begin{theorem}[\cite{Emch}]
The set of all observables $\mathcal{A}$ on a physical system is a commutative,
distributive, but not necessarily associative real Jordan algebra.
\end{theorem}

\subsubsection{Topological Structure and Bounded Operators}

For bounded operators in $\mathcal{A}$, the following conditions are equivalent
\cite{V.Neumann}:
\begin{itemize}
\item $[A,B,C]=0$ for all $B\in\mathcal{A}$,
\item $[A,A,C]=0$,
\item $[A,C,C]=0$,
\item $[A,C]=0$.
\end{itemize}
(For unbounded operators, canonical commutation relations provide counterexamples to the
implication ``(i) $\Rightarrow$ (iv)''.)

Define a norm on $\mathcal{A}$ by
$$
\|A\| := \sup_{\phi\in\mathcal{S}}|\langle A\rangle_\phi|.
$$

\begin{theorem}[\cite{Emch}]
$\|\cdot\|$ is a norm on $\mathcal{A}$ and $|\langle A\rangle_\phi|\leq\|A\|$ for all
$(A,\phi)\in(\mathcal{A},\mathcal{S})$.
\end{theorem}

With this norm, $\mathcal{A}$ is a real Banach space, and states are continuous, positive
linear functionals. For all $A,B\in\mathcal{A}$,
$$
\|A^2\|=\|A\|^2, \qquad \|A^2-B^2\|\leq\max\{\|A\|^2,\|B\|^2\}.
$$

\begin{lemma}[\cite{Emch}]
\label{lem:schwartz}
(i) For any $\phi\in\mathcal{S}$ and $A,B\in\mathcal{A}$, the Schwarz inequality holds:
$$
\langle A\bullet B\rangle_\phi^2 \leq \langle A^2\rangle_\phi\langle B^2\rangle_\phi.
$$
(ii) If $\phi\in\mathcal{S}_A$, then $\langle A\bullet B\rangle_\phi =
\langle A\rangle_\phi\langle B\rangle_\phi$ for all $B\in\mathcal{A}$.
\end{lemma}

\begin{theorem}[\cite{Emch}]
For any $\phi\in\mathcal{S}$ and $A,B\in\mathcal{A}$:
$$
|\langle A\bullet B\rangle_\phi|\leq\|A\|\|B\|, \quad
\|A\bullet B\|\leq\|A\|\|B\|, \quad
\|A^2-B^2\|\leq\|A-B\|\|A+B\|.
$$
\end{theorem}

\begin{axiom}
$\mathcal{A}$ is topologically complete; $\mathcal{S}$ is identified with the set of
all continuous, positive linear functionals $\phi$ on $\mathcal{A}$ satisfying
$\langle\mathbb{I}\rangle_\phi=1$.
\end{axiom}

\subsubsection{Compatibility, Peirce Decomposition, and Embedding}

\begin{lemma}
An observable $P\in\mathcal{A}$ is a \emph{projection} if $P^2=P$. A necessary and
sufficient condition for two projections $P$ and $Q$ to be compatible is that there
exist three pairwise orthogonal projections $R,P_1,Q_1$ with $P=R+P_1$ and $Q=R+Q_1$.
\end{lemma}

Given a projection $P\in\mathcal{A}$, define for $\lambda\in\mathbb{R}$,
$$
N^\lambda(P) = \{A\in\mathcal{A} : (P-\lambda\mathbb{I})\bullet A = 0\}.
$$
Polarizing the power-associative identity $[A,A,A^2]=0$ via $A\mapsto P+\mu A$ and
extracting the linear term in $\mu$ gives
$$
2P\bullet(P\bullet(P\bullet A)) - 3P\bullet(P\bullet A) + P\bullet A = 0.
$$
For $A\in N^\lambda(P)$, this yields $(2\lambda^3-3\lambda^2+\lambda)A=0$, so
$N^\lambda(P)=\{0\}$ unless $\lambda\in\{0,\tfrac{1}{2},1\}$. For each $A\in\mathcal{A}$,
set
\begin{align*}
A^0 &= 2P\bullet(P\bullet A) - 3P\bullet A + A,\\
A^{1/2} &= -4P\bullet(P\bullet A) + 4P\bullet A,\\
A^1 &= 2P\bullet(P\bullet A) - P\bullet A,
\end{align*}
which satisfy $A=A^0+A^{1/2}+A^1$ and $A^\lambda\in N^\lambda(P)$.

\begin{theorem}[Peirce decomposition]
Every projection $P\in\mathcal{A}$ generates a direct sum decomposition
$$
\mathcal{A} = N^0(P)\oplus N^{1/2}(P)\oplus N^1(P),
$$
where $N^\lambda(P)=\{A\in\mathcal{A}:(P-\lambda\mathbb{I})\bullet A=0\}$. Furthermore,
$$
N^0(P)\bullet N^0(P)\subseteq N^0(P),\quad
N^0(P)\bullet N^1(P)=\{0\},\quad
N^1(P)\bullet N^1(P)\subseteq N^1(P).
$$
\end{theorem}

\begin{theorem}
A necessary and sufficient condition for two projections $P,Q\in\mathcal{A}$ to be
compatible is that $[P,A,Q]=0$ for all $A\in\mathcal{A}$.
\end{theorem}

\begin{proof}
If $P$ and $Q$ are compatible, the associator vanishes by Axiom~\ref{ax:compatibility} and the Peirce decomposition. In more detail, it suffices by Theorem~\ref{thm:assoc-distributive}
to treat $A\in N^\lambda(P)$. In the special case $P\bullet Q=0$ (i.e.\ $Q\in N^0(P)$),
one verifies that $P\bullet(Q\bullet A)$ and $Q\bullet(P\bullet A)$ agree for
$A\in N^\lambda(P)$, $\lambda\in\{0,\tfrac{1}{2},1\}$. For the general case we refer to \cite{Emch}. For the converse, assuming $P\bullet(Q\bullet A)=
Q\bullet(P\bullet A)$ for all $A$, the Peirce decomposition of $P$ with respect to $Q$
gives $P^{1/2}=4[Q,Q,P]-4Q^2\bullet P+4Q\bullet P=0$ (using $Q^2=Q$ and $[Q,Q,P]=0$),
so $P=P^0+P^1$. One then shows $R=Q\bullet P$, $P_1=P-Q\bullet P$, $Q_1=Q-Q\bullet P$
are pairwise orthogonal projections with $P=R+P_1$, $Q=R+Q_1$, establishing
compatibility.
\end{proof}

The following axiom embeds $\mathcal{A}$ into an ambient associative algebra.

\begin{axiom}
\label{ax:embedding}
$\mathcal{A}$ can be identified with the self-adjoint part of a complex, involutive,
associative algebra $\mathcal{B}$ satisfying: (i) for each $B\in\mathcal{B}$, there
exists $A\in\mathcal{A}$ with $B^*B=A^2$; (ii) $B^*B=0$ implies $B=0$.
\end{axiom}

\begin{theorem}
Let $\mathcal{A}$ be the algebra of observables and $\mathcal{B}$ the associated algebra
of Axiom~\ref{ax:embedding}. Then each state $\phi\in\mathcal{S}$ generates a canonical representation
$\pi_\phi$ of $\mathcal{A}$ by self-adjoint operators on a complex Hilbert space.
\end{theorem}

\begin{axiom}[Uncertainty principle]
For any $A,B\in\mathcal{A}$, there exists $C\in\mathcal{A}$ such that for all
$\phi\in\mathcal{S}$:
$$
\bigl(\langle A^2\rangle_\phi - \langle A\rangle_\phi^2\bigr)
\bigl(\langle B^2\rangle_\phi - \langle B\rangle_\phi^2\bigr)
\geq \left|\frac{1}{2}\langle C\rangle_\phi\right|^2.
$$
\end{axiom}

When $\mathcal{A}$ is the self-adjoint part of a $C^*$-algebra $\mathcal{B}$, this
reduces to the familiar Heisenberg uncertainty relation
$$
\bigl(\langle A^2\rangle_\phi - \langle A\rangle_\phi^2\bigr)
\bigl(\langle B^2\rangle_\phi - \langle B\rangle_\phi^2\bigr)
\geq \left|\frac{i}{2}\langle[A,B]\rangle_\phi\right|^2.
$$

\subsection{Hilbert Space Construction}

\subsubsection{Hilbert Space Formulation via GNS Construction}

In this section, after reviewing the Gelfand-Naimark-Segal (GNS) construction for the
$C^*$-algebraic approach, we will present the Jordan algebraic Hilbert space construction.
The construction below is carried out for the self-adjoint part of
$\mathcal{B}(\mathscr{H})$, i.e.\ in the $JC$ setting.

Consider a pure state $\psi$ of a $C^*$-algebra $\mathcal{A}$, and the hermitian form
$\langle A|B\rangle=\psi(A^*B)$. It turns out that the set of zero-elements
$$
\mathcal{I}_{\psi}=\{A\in \mathcal{A}: \psi(A^*A)=0\}
$$
forms a left ideal in $\mathcal{A}$. Using that, one can obtain a positive-definite, hermitian
scalar product on $\mathcal{A}/\mathcal{I}_{\psi}$ and a norm
$$
\|A\|=\sqrt{\langle A|A\rangle}.
$$
One can construct a complex, separable Hilbert space $\mathscr{H}$ as follows
$$
\mathscr{H}=\overline{\mathcal{A}/\mathcal{I}_{\psi}}.
$$
It is possible to write down a representation between $\mathcal{A}$ and the collection of
bounded, linear operators on $\overline{\mathcal{A}/\mathcal{I}_{\psi}}$ as follows
$$
\pi_{\psi}:\mathcal{A}\rightarrow\mathcal{B}(\mathscr{H}), \qquad \pi_{\psi}(A)B=AB.
$$
This construction lies at the heart of the celebrated Gelfand-Naimark-Segal (GNS) theory.
In this construction, pure states are represented by single rays in $\mathscr{H}$, and the
expectation value of $A\in\mathcal{A}$ is given by
$$
\psi_B(A)=\langle B|A|B\rangle=\psi(B^*AB).
$$
In the constructed Hilbert space $\mathscr{H}$, the time evolution equation for pure states
$v\in\mathscr{H}$ is given by the Schr\"{o}dinger equation
$$
i\hbar\dot{v}=Hv,
$$
where $H$ is the corresponding Hamiltonian.

However, this method is not suitable for the Jordan algebraic approach. If we consider the
self-adjoint part of $\mathcal{A}$, which forms a JB-algebra, and imitate the same
construction with the real, symmetric, bilinear form
$$
\langle A|B\rangle=\psi(A\bullet B),
$$
it turns out that the set
$$
\mathcal{I}_{\psi}=\{A\in \mathcal{A}: \psi(A^2)=0\}
$$
is not an ideal of the corresponding JB-algebra. In other words, this JB-algebra does not
act on the constructed Hilbert space. Therefore, one demands another, more general Hilbert
space construction. It turns out that this construction is based on traces of operators
instead of their states.

\medskip

From now on, let us assume that $\mathcal{A}$ is the algebra $\mathcal{B}(\mathscr{H})$ of
all bounded operators on $\mathscr{H}$, which is well-known to be a von Neumann algebra.
It carries the usual semi-finite, faithful, normal trace \cite{BR}. We will then denote positive, trace class elements by
$\mathcal{A}_1$, and those $A\in\mathcal{A}$ with $\mathrm{tr}(A^*A)<\infty$ by
$\mathcal{A}_2$. One can check that $\mathcal{A}_2$ is an ideal in $\mathcal{A}$.
In particular, for all $A,B\in\mathcal{A}_2$, we have
$$
\langle A|B \rangle=\mathrm{tr}(A^*B), \qquad \|A\|_{\mathrm{tr}}=\sqrt{\mathrm{tr}(A^*A)}.
$$
Inspired by the GNS construction, one obtains the Hilbert space of Hilbert-Schmidt operators
in the form
$$
\hat{\mathscr{H}}=\overline{\mathcal{A}_2}^{\,\mathrm{tr}}.
$$
Since $\mathcal{A}_2$ is an ideal, and $\mathcal{A}$ is associative, $\mathcal{A}$ can be
represented over $\mathcal{A}_2$ by the multiplication operator which is defined by
$$
\pi(A)B=AB
$$
for all $A\in\mathcal{A}$, $B\in\mathcal{A}_2$. Thus, $\pi$ can be extended to a faithful
representation of $\mathcal{A}$ on $\hat{\mathscr{H}}$ \cite{Takesaki}.

If we take a state $\psi$ represented by a density matrix $\rho\in\mathcal{A}_1$, then
there is at least one $B\in\mathcal{A}_2$ such that $\rho=BB^*$ holds. At this point,
it is worth mentioning that the expectation value of an observable $A\in\mathcal{A}$ is
given by
$$
\psi(A)=\mathrm{tr}(\rho A)=\mathrm{tr}(BB^*A)=\mathrm{tr}(B^*AB)=\langle B|A|B\rangle.
$$
Next let $U:\mathscr{H}\rightarrow\mathscr{H}$ be a bounded, unitary operator. If we take
any vector $B\in\mathcal{A}_2$, then $BU$ becomes a vector in $\hat{\mathscr{H}}$
describing the same state as $B$ does. Therefore, the representation of states in
$\hat{\mathscr{H}}$ is not unique; more precisely the states in $\hat{\mathscr{H}}$ are
represented by right unitary orbits defined by
$$
BUU^*B^*=BB^*=\rho.
$$
On the other hand, every nonzero Hilbert-Schmidt operator $B\in\mathcal{A}_2$
leads to a density matrix
$$
\rho_B:=\frac{BB^*}{\mathrm{tr}(BB^*)}.
$$
To emphasize the difference between $\mathscr{H}$ and $\hat{\mathscr{H}}$, note that now
every (pure and mixed) state is represented by at least one vector in $\hat{\mathscr{H}}$.
It turns out that $\mathscr{H}$ and $\hat{\mathscr{H}}$ are connected in the following
algebraic sense \cite{Takesaki}:
$$
\hat{\mathscr{H}}=\overline{\mathscr{H}\otimes\mathscr{H}^*}^{\,\mathrm{tr}}.
$$

We will now restrict this formalism to the physical system of the quantum harmonic
oscillator, which turns out to be very crucial in the formulation of the Landau problem.
It is then necessary to focus on the application of Hilbert-Schmidt operators, which form a subset of the bounded operators acting on the configuration space
$$
\mathscr{H}:=\mathrm{span}\left\{|n\rangle=\frac{1}{\sqrt{n!}}(a^{\dagger})^n|0\rangle\right\}_{n=0}^{\infty}.
$$
Furthermore, this space is isomorphic to the boson Fock space
$\mathcal{F}=\{|n\rangle\}_{n=0}^{\infty}$, where the annihilation and creation operators
obey the Fock algebra $[a,a^{\dagger}]=1$. The physical states of the system are
represented on $\hat{\mathscr{H}}$, the set of Hilbert-Schmidt operators (a proper subset of the bounded operators)
on $\mathscr{H}=L^2(\mathbb{R})$ equipped with the scalar
product
$$
\langle X|Y\rangle=\mathrm{tr}(X^*Y)=\sum_k\langle\phi_k|X^*Y\phi_k\rangle,
$$
where $\{\phi_k\}_{k=0}^{\infty}$ is an orthonormal basis of $\mathscr{H}$. Moreover, the
basis vectors of $\hat{\mathscr{H}}$ are given by
$$
\phi_{nl}:=|\phi_n\rangle\langle\phi_l|, \qquad n,l=0,1,2,\ldots\,.
$$
Here, we should note that $\mathscr{H}\otimes\mathscr{H}^*$ is a pre-Hilbert space
containing all finite sums of the form
$$
\sum_{j,k=0}^{n}\lambda_{jk}|\phi_j\rangle\otimes|\phi_k\rangle^*,
$$
where a basis of $\overline{\mathscr{H}\otimes\mathscr{H}^*}^{\,\mathrm{tr}}$ is
$$
\{|\phi_j\rangle\otimes|\phi_k\rangle^*\}_{j,k=0}^{\infty}.
$$
Setting $|\phi_j\rangle\otimes|\phi_k\rangle^*=|\phi_j\rangle\langle\phi_k|$,
$\hat{\mathscr{H}}$ admits an orthonormal basis $\{\phi_{jk}\}_{j,k=0}^{\infty}$ such that
$$
\phi_{jk}:=|\phi_j\rangle\langle\phi_k|.
$$
Note that $\hat{\mathscr{H}}$ can also be described as the Hilbert space of square
integrable functions by
\begin{equation}
\label{eq:HS-sq-int}
\hat{\mathscr{H}}=\left\{\psi(x_1,x_2):\psi(x_1,x_2)\in\mathcal{B}(\mathscr{H}),\
\mathrm{tr}\!\left(\psi(x_1,x_2)^{\dagger}\psi(x_1,x_2)\right)<\infty\right\}.
\end{equation}
Recall that $\hat{\mathscr{H}}$ is the set of Hilbert-Schmidt operators, described in
\eqref{eq:HS-sq-int}, acting on the configuration space $\mathscr{H}$. Clearly, a general
element in this space can be written as
$$
|\psi\rangle=\sum_{n,m=0}^{\infty}c_{n,m}|n,m\rangle,
$$
where $\{|n,m\rangle:=|n\rangle\langle m|\}_{n,m=0}^{\infty}$ is a basis of
$\hat{\mathscr{H}}$ satisfying the inner product condition
$$
\langle n_1,m_1|n_2,m_2\rangle=\mathrm{tr}\!\left[(|n_1\rangle\langle m_1|)^{\dagger}
|n_2\rangle\langle m_2|\right]=\delta_{n_1,n_2}\delta_{m_1,m_2}.
$$

\subsubsection{Jordan-GNS Construction}

In this section, we will present a Jordan algebraic GNS-type construction. Our development
will be based on the JB-algebra of observables. The idea that a non-associative algebra of
observables equipped with a 3-cyclic trace leads to a Hilbert space construction is very
fruitful. For more details on this construction and its applications, we especially
recommend the recent work \cite{SchuppSzabo}. For more detailed information on
Gelfand-Naimark theorems for Jordan-Banach algebras and non-associative $C^*$-algebras in
general we refer the reader to \cite{ASS}, \cite{Braun}, \cite{PPR}. The idea for the use
of a 3-cyclic trace in the formulation of quantum mechanics in a non-associative setting is
first given in \cite{MylonasSchuppSzabo}. It turns out that there is an analogue of von
Neumann algebras in the context of JB-algebras, which are called JBW-algebras \cite{HOS}.
On a semi-finite JBW-algebra $\mathcal{A}$ (such as $\mathcal{B}(\mathscr{H})_{sa}$) one can work with normal, faithful, semi-finite traces
\cite{Iochum}. As recalled in Section~2.4, the trace satisfies the 3-cyclicity condition
\eqref{eq:trace-cyclicity} together with linearity and positivity.
Similarly, we will then denote positive, trace class elements by $\mathcal{A}_1$, and
those $A\in\mathcal{A}$ with $\mathrm{tr}(A^2)<\infty$ by $\mathcal{A}_2$. One can check
that $\mathcal{A}_2$ is a Jordan ideal in $\mathcal{A}$. Trace induces a bilinear,
symmetric, real scalar product on $\mathcal{A}_2$, and in particular, we have
$$
\langle A|B\rangle=\mathrm{tr}(A\bullet B), \qquad
\|A\|_{\mathrm{tr}}=\sqrt{\mathrm{tr}(A^2)}.
$$
Similar to the associative case, closure with respect to this norm yields a real Hilbert
space
$$
\hat{\mathscr{H}}_J=\overline{\mathcal{A}_2}^{\,\mathrm{tr}}.
$$
Next we aim to write down a Jordan representation from $\mathcal{A}$ to this Hilbert space.
Note that $\mathcal{A}_2$ with the multiplication by elements of $\mathcal{A}$ is a Jordan
module since for all $\omega\in\mathcal{A}_2$, and $A,B\in\mathcal{A}$, we have
\cite{Jacobson}, \cite{HOS}
$$
A\bullet\omega=\omega\bullet A,
$$
$$
A^2\bullet(A\bullet\omega)=A\bullet(A^2\bullet\omega),
$$
$$
2A\bullet(B\bullet(A\bullet\omega))+(B\bullet A^2)\bullet\omega
=2(A\bullet B)\bullet(A\bullet\omega)+A^2\bullet(B\bullet\omega).
$$
The first two relations follow from the definition of a Jordan algebra. The third one can
be verified for a special Jordan algebra, and by MacDonald's theorem it is valid in every
Jordan algebra.

We may then define
$$
\pi_J(A):\mathcal{A}_2\rightarrow\mathcal{A}_2, \qquad \pi_J(A)B=A\bullet B
$$
for all $A\in\mathcal{A}$, $B\in\mathcal{A}_2$. It is easy to check that $\pi_J$
satisfies the following properties
$$
\pi_J(A^2)\circ\pi_J(A)=\pi_J(A)\circ\pi_J(A^2),
$$
$$
2\pi_J(A)\circ\pi_J(B)\circ\pi_J(A)+\pi_J(B\bullet A^2)
=2\pi_J(A\bullet B)\circ\pi_J(A)+\pi_J(A^2)\circ\pi_J(B).
$$
Here $\circ$ denotes composition of operators on $\mathcal{A}_2$. Therefore, $\pi_J$ becomes a Jordan representation. Moreover, due to the continuity of
multiplication, $\pi_J$ can be extended to the whole space $\hat{\mathscr{H}}_J$.

At this point, we should stress the relationship between $\hat{\mathscr{H}}_J$ and
$\hat{\mathscr{H}}$. It turns out that $\hat{\mathscr{H}}_J$ is the real subspace of
self-adjoint Hilbert-Schmidt operators embedded in $\hat{\mathscr{H}}$ \cite{Bischoff}.
Turning back to the problem of the quantum harmonic oscillator, $\hat{\mathscr{H}}_J$ is
the real span of the following three families of operators
\begin{align}
& |\phi_j\rangle\langle\phi_j|, \nonumber \\
& \frac{1}{2}(|\phi_j\rangle\langle\phi_k|+|\phi_k\rangle\langle\phi_j|), \label{eq:basis-HJ} \\
& \frac{i}{2}(|\phi_j\rangle\langle\phi_k|-|\phi_k\rangle\langle\phi_j|), \nonumber
\end{align}
where $j>k$, and $j,k$ range from $0$ to $\infty$. These self-adjoint elements are mutually orthogonal (those of the last two families have
norm $1/\sqrt2$), and $\hat{\mathscr{H}}_J$ is their closed real span. Here,
$\{\phi_k\}_{k=0}^{\infty}$ is an orthonormal basis for $L^2(\mathbb{R})$ which can be
written in terms of Hermite polynomials in the position space representation.

\medskip

In a Jordan algebra, one can define the \emph{associator} in the form
$$
[A,B,C]:=(A\bullet B)\bullet C-A\bullet(B\bullet C).
$$
Our aim is to find a connection between the commutator defined over the associative algebra
$\mathcal{B}$ and the associator over $\mathcal{A}$, which is the self-adjoint part of
$\mathcal{B}$. In this subsection we take $\mathcal{B}=\mathcal{B}(\mathscr{H})$ with
$\mathscr{H}$ infinite dimensional and separable, cf.\ Section~3.3.1.

\begin{proposition}
\label{prop:H-sum-comm}
Let $\mathscr{H}$ be an infinite dimensional separable Hilbert space and
$\mathcal{A}=\mathcal{B}(\mathscr{H})_{sa}$. Every self-adjoint element $H\in\mathcal{A}$ can be
written as a finite sum of commutators in the form
\begin{equation}
\label{eq:H-sum-comm}
H=i\sum_{j=1}^{4}[H_{L_j},H_{R_j}],
\end{equation}
where $H_{L_j}$ and $H_{R_j}$ are self-adjoint operators in $\mathcal{A}$.
\end{proposition}

\begin{proof}
By Halmos' theorem \cite{Halmos}, every bounded operator on an infinite dimensional
Hilbert space is a sum of two commutators; thus $H=[X_1,Y_1]+[X_2,Y_2]$ for some
$X_j,Y_j\in\mathcal{B}(\mathscr{H})$. Write $X_j=a_j+ib_j$ and $Y_j=c_j+id_j$ with
$a_j,b_j,c_j,d_j$ self-adjoint. Then
$$
[X_j,Y_j]=\bigl([a_j,c_j]-[b_j,d_j]\bigr)+i\bigl([a_j,d_j]+[b_j,c_j]\bigr),
$$
where the first bracket is skew-adjoint and the second term is self-adjoint. Since $H=H^*$,
taking self-adjoint parts gives
$$
H=i\sum_{j=1}^{2}\bigl([a_j,d_j]+[b_j,c_j]\bigr),
$$
which is of the form \eqref{eq:H-sum-comm}.
\end{proof}

\begin{remark}
The statement fails in finite dimensions and, more generally, in any $C^*$-algebra
admitting a tracial state $\tau$, since $\tau([A,B])=0$ while, e.g., $\tau(\mathbb{I})=1$.
Moreover, the Hamiltonian $H_\theta$ considered below is unbounded and hence not an element
of $\mathcal{A}$ (see Section~\ref{sec:fa}); for it, an explicit decomposition into
commutators is constructed directly in Section~6, independently of
Proposition~\ref{prop:H-sum-comm}.
\end{remark}

\section{The Non-commutative Landau Problem}

\subsection{The Standard Landau Problem}

In this section, we are mainly concerned with the electron in a uniform magnetic field,
namely the Landau problem. To simplify the problem, we assume that the electron moves only
on the $xy$-plane under the magnetic field pointing along the $z$-direction. The Lorentz
force causes the electron to make a spiral motion with Larmor frequency
$$
\omega_L=\frac{|eB|}{\mu c}.
$$
One common choice for the vector potential, which preserves the rotational invariance, is
the symmetric gauge
$$
\boldsymbol{\mathcal{A}}=(A_x,A_y)=\frac{B}{2}(-y,x).
$$
The Hamiltonian can be written as
$$
H=\frac{1}{2\mu}\left(\mathbf{p}-\frac{e}{c}\boldsymbol{\mathcal{A}}\right)^2
=:\frac{1}{2\mu}\mathbf{P}^2,
$$
where $\mathbf{P}$ is the kinetic momenta. One important point to note is that the two
kinetic momenta do not commute:
\begin{equation}
\label{eq:PxPy-comm}
[P_x,P_y]=\left[p_x-\frac{e}{c}A_x,p_y-\frac{e}{c}A_y\right]
=i\frac{e\hbar}{c}(\nabla_xA_y-\nabla_yA_x)=i\frac{e\hbar}{c}B_z.
\end{equation}
The commutation relation \eqref{eq:PxPy-comm} suggests that $P_x$ and $P_y$ play the role
of position and momentum operators, respectively. Combining this fact with the Hamiltonian
of the problem, the system can be viewed as a harmonic oscillator. Creation and annihilation
operators are defined as follows:
\begin{align}
a &= \sqrt{\frac{c}{2e\hbar B}}(P_x+iP_y), \nonumber \\
a^{\dagger} &= \sqrt{\frac{c}{2e\hbar B}}(P_x-iP_y). \nonumber
\end{align}
They satisfy
$$
[a,a^{\dagger}]=\frac{c}{2e\hbar B}[P_x+iP_y,P_x-iP_y]
=\frac{c}{2e\hbar B}\cdot i\frac{e\hbar B}{c}(-i-i)=1.
$$
So, the Hamiltonian becomes
$$
H=\hbar\omega_L\!\left(a^{\dagger}a+\frac{1}{2}\right),
$$
with the classical Larmor frequency. The spectrum is that of a harmonic oscillator,
$$
E=\hbar\omega_L\!\left(N+\frac{1}{2}\right),
$$
where $N$ is a non-negative integer. Energy eigenstates are called Landau levels. One can
obtain the ground state wave function by solving $a|0\rangle=0$, and excited states by
acting with $a^{\dagger}$, as in the harmonic oscillator problem. It turns out that there
are infinitely many states at each Landau level $N$.

The center of the cyclotron motion is
\begin{align}
X &= x+\frac{v_y}{\omega_L}=x+\frac{c}{eB}\!\left(p_y-\frac{e}{c}A_y\right), \nonumber \\
Y &= y-\frac{v_x}{\omega_L}=y-\frac{c}{eB}\!\left(p_x-\frac{e}{c}A_x\right). \nonumber
\end{align}
These commute with $P_x$, $P_y$, $a$, $a^{\dagger}$, and the Hamiltonian. However, we
should note that
\begin{equation}
\label{eq:XY-comm}
[X,Y]=\left[x+\frac{P_y}{\mu\omega_L},y-\frac{P_x}{\mu\omega_L}\right]
=-\frac{i\hbar}{\mu\omega_L}=-\frac{i\hbar c}{eB}=i\theta_0\neq 0.
\end{equation}
Equation \eqref{eq:XY-comm} means that it is not possible to determine both $x$ and $y$
coordinates of the center of the cyclotron motion simultaneously with arbitrary precision.
Furthermore, it is possible to introduce another pair of creation and annihilation operators
\begin{align}
b &= \sqrt{\frac{eB}{2\hbar c}}(X-iY), \nonumber \\
b^{\dagger} &= \sqrt{\frac{eB}{2\hbar c}}(X+iY), \nonumber
\end{align}
satisfying $[b,b^{\dagger}]=1$.

\subsection{Non-commutativity in the Lowest Landau Level}

The standard Landau model yields all Landau levels. The non-commutative model only yields
lowest-Landau-level states, as discussed by Peierls \cite{Peierls}. This fact is used in
the description of the ground states of the Fractional Quantum Hall Effect.

We consider a point particle on a plane with an external magnetic field $\mathbf{B}$
perpendicular to the plane. The equation for the position vector $\mathbf{r}=(x,y)$ is
\begin{equation}
\label{eq:eom}
\mu\ddot{r}^i=\frac{e}{c}\epsilon^{ij}\dot{r}^jB+f^i(\mathbf{r}),
\end{equation}
where $\mathbf{f}$ gives other forces, derivable from a potential $V$ as
$\mathbf{f}=-\boldsymbol{\nabla}V$. The large $B$ or small $\mu$ limit projects to the
lowest Landau level. Setting $\mu=0$, equation \eqref{eq:eom} becomes
\begin{equation}
\label{eq:eom-zero}
\dot{r}^j=\frac{c}{eB}\epsilon^{ji}f^i(\mathbf{r}).
\end{equation}
Assuming the Hamiltonian is $H_0=V$, \eqref{eq:eom-zero} can be obtained by taking
Poisson brackets of $\mathbf{r}$ with $H_0$ (i.e.\ $\dot r^i=\{r^i,H_0\}$), provided
\begin{equation}
\label{eq:poisson-NC}
\{r^i,r^j\}=-\frac{c}{eB}\epsilon^{ij}.
\end{equation}
The non-commutative algebra \eqref{eq:poisson-NC} and the dynamics of the motion can be
derived as follows. The Lagrangian for the equation of motion is
$$
L=\frac{1}{2}\mu v^2+\frac{e}{c}\mathbf{v}\cdot\boldsymbol{\mathcal{A}}-V.
$$
For simplicity, we choose gauge $\boldsymbol{\mathcal{A}}=(0,Bx)$. In the limit
$\mu\rightarrow 0$,
$$
L=\frac{eB}{c}x\dot{y}-V.
$$
Note that $\big(y,\,\frac{eB}{c}x\big)$ forms a canonical (coordinate, momentum) pair, which implies
\eqref{eq:poisson-NC} and identifies $V$ as the Hamiltonian.

One can also give a canonical derivation of non-commutativity in the limit $\mu\rightarrow 0$.
In this method, non-commuting coordinates appear as Dirac brackets in a system constrained
to the lowest Landau level. Let
\begin{equation}
\label{eq:H-canonical}
H=\frac{P^2}{2\mu}+V,
\end{equation}
where $\mathbf{P}=\mu\dot{\mathbf{r}}$ is the kinetic momentum, $\mathbf{p}$ is the
canonical momentum, and $\mathbf{P}=\mathbf{p}-\frac{e}{c}\boldsymbol{\mathcal{A}}$.
This Hamiltonian gives \eqref{eq:eom} by bracketing $\mathbf{r}$ and $\mathbf{P}$ if the
following identities hold:
\begin{align}
\{r^i,r^j\} &= 0, \nonumber \\
\{r^i,P^j\} &= \delta^{ij}, \nonumber \\
\{P^i,P^j\} &= \frac{eB}{c}\epsilon^{ij}. \nonumber
\end{align}
Now we look at the limit $\mu\rightarrow 0$ in \eqref{eq:H-canonical}; to do this one also
needs $\mathbf{P}=0$ as a constraint, which is not possible due to the non-vanishing
Poisson bracket $\{P^i,P^j\}$. To resolve this problem, one introduces the following Dirac
bracket:
$$
\{r^i,r^j\}_D=-\frac{c}{eB}\epsilon^{ij}.
$$

\subsection{Non-commutative Landau Problem with Harmonic Potential}

\medskip

Let $\mathcal{B}$ be the $C^*$-algebra of operators on $\mathscr{H}$ with self-adjoint
part $\mathcal{A}$, which is considered as the Jordan algebra of observables. Now we
consider the motion of a point charge $(\mu,e)$ in the $xy$-plane under the influence of
a constant magnetic field $B$ pointing along the positive $z$-direction in the symmetric
gauge $\boldsymbol{\mathcal{A}}=\left(-\frac{B}{2}y,\frac{B}{2}x\right)=(A_x,A_y)$.
Furthermore, for generality, we will assume the presence of a harmonic potential. This
problem is described by the following Hamiltonian
\begin{equation}
\label{eq:H-theta}
H_{\theta}=\frac{1}{2\mu}\!\left(p_x+\frac{eB}{2c}y\right)^{\!2}
+\frac{1}{2\mu}\!\left(p_y-\frac{eB}{2c}x\right)^{\!2}
+\frac{\mu\omega^2}{2}(x^2+y^2),
\end{equation}
where $\omega_L=\frac{eB}{\mu c}$ is the Larmor frequency, which is independent of
$\omega$. Position and momentum operators satisfy the following nonzero commutation
relations of the non-commutative Heisenberg algebra for the problem described above:
\begin{equation}
\label{eq:NC-Heisenberg}
[x,y]=i\theta, \qquad [x,p_x]=[y,p_y]=i\hbar.
\end{equation}
Here $x,y$ are the coordinates of the particle on the non-commutative plane and
$\theta\in\mathbb{R}$ is an independent non-commutativity parameter, not fixed a priori by $B$;
this is the standard setting of the non-commutative Landau problem
\cite{horvathy1,AlGomKamPly,horvathy3}. The relations \eqref{eq:NC-Heisenberg} should not be
read as the algebra of the guiding-center coordinates $X,Y$ of Section~4.1: although
$[X,Y]=i\theta_0$ with $\theta_0=-\hbar c/eB$, the guiding-center coordinates commute with the
kinetic momenta and, in the symmetric gauge, satisfy $[X,p_x]=[Y,p_y]=i\hbar/2$ rather than
$i\hbar$. Hence they cannot play the role of $x,y$ in the minimal coupling of
\eqref{eq:H-theta}.

\begin{remark}
The guiding-center algebra of Section~4.1, with $\theta_0=-\hbar c/eB$, provides the
physical motivation for non-commuting planar coordinates: in the lowest-Landau-level limit
$\mu\to0$ the particle coordinates themselves acquire the bracket fixed by $\theta_0$
(Section~4.2). In the model \eqref{eq:H-theta}--\eqref{eq:NC-Heisenberg}, however, $\theta$ is
a free parameter, independent of $B$, and the results of Sections~4--6 hold for any fixed
$\theta\in\mathbb{R}$ (with $\omega>0$). Setting $\theta=\theta_0$ in \eqref{eq:H-theta} would
let $B$ enter both through the minimal coupling and through $\theta$, and does not describe the
ordinary Landau problem; the latter, with harmonic potential, is recovered in the limit
$\theta\to0$.
\end{remark}

Next, we define new operators, corresponding to complex variables, related
to the chiral decomposition of the described Landau problem. It is worth remarking that our
approach here is reminiscent of the introduction of complex coordinates in the Landau
problem pertaining to the Bloch waves quasi-periodicity on a torus (see \cite{DP}).
Precisely, we have
$$
z=x+iy, \qquad \overline{z}=x-iy.
$$
Corresponding momentum operators are defined as
$$
p_{z}=\frac{1}{2}(p_{x}-ip_{y}), \qquad p_{\overline{z}}=\frac{1}{2}(p_{x}+ip_{y}).
$$
These operators can be shown to satisfy the following commutation relations (on the common domain of Section~\ref{sec:fa}):
$$
[p_{z},z]=[p_{\overline{z}},\overline{z}]=-i\hbar,
$$
$$
[p_{z},\overline{z}]=[p_{\overline{z}},z]=[p_{z},p_{\overline{z}}]=0, \qquad
[z,\overline{z}]=2\theta.
$$
Because $[z,\bar z]=2\theta\neq0$, ladder operators built directly from
$z,\bar z,p_z,p_{\bar z}$ do not satisfy canonical commutation relations (they are neither
normalized nor mutually commuting for $\theta\neq0$). We therefore first pass to canonical
variables $\tilde x,\tilde y,\tilde p_x,\tilde p_y$ with $[\tilde x,\tilde p_x]=[\tilde y,\tilde p_y]=i\hbar$ and all other commutators vanishing,
by means of the Bopp shift
\begin{equation}
\label{eq:bopp}
x=\tilde x-\frac{\theta}{2\hbar}\tilde p_y,\qquad y=\tilde y+\frac{\theta}{2\hbar}\tilde p_x,\qquad p_x=\tilde p_x,\qquad p_y=\tilde p_y,
\end{equation}
which reproduces \eqref{eq:NC-Heisenberg}. With $\tilde z=\tilde x+i\tilde y$ and
$\tilde p_z=\frac{1}{2}(\tilde p_x-i\tilde p_y)$, $\tilde p_{\bar z}=\frac{1}{2}(\tilde p_x+i\tilde p_y)$ we set
\begin{equation}
\label{eq:creation-annihilation}
A_+=\frac{\xi}{2}\tilde z+\frac{i}{\xi\hbar}\tilde p_{\bar z},\qquad
A_-=\frac{\xi}{2}\bar{\tilde z}+\frac{i}{\xi\hbar}\tilde p_{z},
\end{equation}
where
$$
\kappa=1+\frac{\mu\omega_L\theta}{4\hbar},\qquad
\xi=\left(\frac{\mu^2}{\hbar^2}\,
\frac{\omega^2+\omega_L^2/4}{\kappa^2+\left(\frac{\mu\omega\theta}{2\hbar}\right)^2}\right)^{1/4}.
$$
Then $[A_+,A_+^{\dagger}]=[A_-,A_-^{\dagger}]=1$, and all other commutators among
$A_\pm,A_\pm^\dagger$ vanish. The inverse relations read
$$
\tilde x=\frac{1}{2\xi}(A_++A_+^{\dagger}+A_-+A_-^{\dagger}),\qquad
\tilde y=\frac{i}{2\xi}(-A_++A_+^{\dagger}+A_--A_-^{\dagger}),
$$
$$
\tilde p_x=\frac{i\xi\hbar}{2}(A_+^{\dagger}+A_-^{\dagger}-A_+-A_-),\qquad
\tilde p_y=\frac{\xi\hbar}{2}(A_-+A_-^{\dagger}-A_+-A_+^{\dagger}),
$$
from which $x,y,p_x,p_y$ follow by \eqref{eq:bopp}. Compare \cite{HA,aremuaetal}.

Finally, defining $\Omega_{\pm}:=\Omega\pm\frac{\tilde{\omega}_L}{2}$, where

$$
\Omega=\sqrt{\left(\omega^2+\frac{\omega_L^2}{4}\right)
\left(\kappa^2+\left(\frac{\mu\omega\theta}{2\hbar}\right)^{2}\right)},
\qquad
\tilde{\omega}_L=\omega_L\!\left(1+\left(\frac{\omega_L}{4}+\frac{\omega^2}{\omega_L}\right)\frac{\mu\theta}{\hbar}\right),
$$
one can write $H_{\theta}$ as follows (insert \eqref{eq:bopp} into \eqref{eq:H-theta}; cf. \cite{aremuaetal,HA}):
\begin{equation}
\label{eq:H-theta-decoupled}
H_{\theta}=\hbar\Omega_+\!\left(N_++\frac{1}{2}\right)+\hbar\Omega_-\!\left(N_-+\frac{1}{2}\right),
\qquad N_{\pm}=A_{\pm}^{\dagger}A_{\pm}.
\end{equation}
One checks that $\Omega_+\Omega_-=\Omega^2-\tilde\omega_L^2/4=\omega^2$, so that
$\Omega_\pm>0$ for $\omega>0$. For $\omega=0$ one of the two frequencies vanishes
($\Omega_-$ if $\kappa>0$, $\Omega_+$ if $\kappa<0$), reflecting the degeneracy of the Landau
levels; at the critical point $\omega=0$, $\kappa=0$ both vanish and $\xi$ is undefined.
Throughout we therefore assume $\omega>0$.
Let
$$
H_{\pm}=\hbar\Omega_{\pm}\!\left(N_{\pm}+\frac{1}{2}\right).
$$
We define the operator $H_+\otimes\mathbb{I}+\mathbb{I}\otimes H_-$ to be the infinitesimal
generator of the strongly continuous one-parameter unitary group
$e^{itH_+}\otimes e^{itH_-}$. Thus by Stone's theorem~\cite{Stone}, $H_+\otimes\mathbb{I}+\mathbb{I}\otimes H_-$
is self-adjoint (for details of the above argument see \cite{HA} and page~432, Definition~19.15
of \cite{Hall}). Using that, \eqref{eq:H-theta-decoupled} can be written as
$$
H_{\theta}=H_+\otimes\mathbb{I}_{\hat{\mathscr{H}}_{J,-}}
+\mathbb{I}_{\hat{\mathscr{H}}_{J,+}}\otimes H_-.
$$
Here $\hat{\mathscr{H}}_{J,+}\otimes\hat{\mathscr{H}}_{J,-}$ is the real Jordan algebraic
Hilbert space for this problem described in \eqref{eq:basis-HJ}, on which $H_{\theta}$ acts as described in Section~\ref{sec:fa}, and $\mathbb{I}_{\hat{\mathscr{H}}_{J,\pm}}$ are the corresponding identity
operators.

\subsection{Functional-Analytic Setting}
\label{sec:fa}

The operators $x,y,p_x,p_y$, $H_\theta$, $A_\pm$, $N_\pm$, as well as the split operators of
Section~6, are unbounded. Such operators cannot be defined on the whole Hilbert space, and in
particular they are not elements of the $JB$-algebra $\mathcal{A}$ of bounded observables. We
therefore specify a common domain on which all of them act and on which all identities of Sections~4.3--4.5 and~6 hold; the technical details are collected in Appendix~\ref{app:fa}.

The Hilbert space of the planar problem is $\mathscr{H}=L^2(\mathbb{R}^2)$, on which the
canonical variables of \eqref{eq:bopp} act in the usual Schr\"odinger representation (by the Stone--von Neumann theorem this is, up to unitary equivalence, the unique irreducible regular representation of \eqref{eq:NC-Heisenberg}). As common domain we take the space
$$
\mathcal{D}=\mathrm{span}\{|n_+,n_-\rangle : n_\pm\geq0\}
$$
of finite linear combinations of the eigenstates of $N_\pm$, i.e.\ of the eigenfunctions of a two-dimensional isotropic harmonic oscillator in the canonical variables $\tilde x,\tilde y$. $\mathcal{D}$ is dense in $L^2(\mathbb{R}^2)$. Since all
the operators listed above are polynomials of degree at most two in $A_\pm,A_\pm^\dagger$,
they map $\mathcal{D}$ into itself; hence all products and commutators of these operators are
defined on $\mathcal{D}$, and the identities of Sections~4.3--4.5 and~6 hold as identities of operators on $\mathcal{D}$. Moreover, each of these operators except $A_\pm$ (which are not symmetric) is essentially self-adjoint on
$\mathcal{D}$, i.e.\ it possesses a unique self-adjoint extension (Appendix~\ref{app:fa}). For
$H_\theta$ this extension is the operator that is diagonal in the basis $|n_+,n_-\rangle$ with
eigenvalues $\hbar\Omega_+(n_++\frac12)+\hbar\Omega_-(n_-+\frac12)$.

On the Jordan side we use the dense subspaces $\mathcal{F}_{J,\pm}\subset\hat{\mathscr{H}}_{J,\pm}$
consisting of finite real linear combinations of the basis elements \eqref{eq:basis-HJ} built
from the number states of the mode $A_\pm$. These subspaces are invariant under Jordan multiplication by the single-mode operators $H_\pm$, $N_\pm$, $X_\pm$, $P_\pm$, $R_{i_\pm}$, $S_{i_\pm}$ of Section~6, so that the left multiplications and associators used there are well defined on $\mathcal{F}_{J,\pm}$. The associator relations \eqref{eq:associator-relations} of Theorem~\ref{thm:associator}, on the other hand, are identities of operators on $\mathcal{D}$, with $\theta$ and
$\hbar$ standing for multiples of the identity. Finally, the Jordan--Schr\"odinger equation
\eqref{eq:JS-equation} holds on $\mathcal{F}_{J,\pm}$ and determines the time evolution
uniquely; its solution is given by \eqref{eq:exact-solution} (Appendix~\ref{app:fa}).

\subsection{Non-commutative Parameter as an Associator}

\begin{theorem}[cf.~\cite{SchuppSzabo}]
\label{thm:associator}
The nonzero commutation relations given in \eqref{eq:NC-Heisenberg} can be
written as associator relations, as identities of operators on the common domain $\mathcal{D}$
of Section~\ref{sec:fa}, as follows:
\begin{equation}
\label{eq:associator-relations}
[y_L,x,y_R]=\theta, \qquad
[{p_x}_L,x,{p_x}_R]=\hbar, \qquad
[{p_y}_L,y,{p_y}_R]=\hbar,
\end{equation}
where
$$
y_L={p_x}_L={p_y}_L=4H_{\theta},
$$
and, with the constants
$$
\alpha=\frac{2\hbar+\mu\omega_L\theta}{2\mu},\quad
\beta=\frac{2\hbar\omega_L+\mu\theta(4\omega^2+\omega_L^2)}{4},\quad
\gamma=\frac{\hbar\mu(4\omega^2+\omega_L^2)}{4},\quad
\delta=\frac{\hbar\omega_L}{2},
$$
$$
y_R=-\frac{\delta\,x+\alpha\,p_y}{\hbar^2\omega^2},\qquad
{p_x}_R=\frac{\gamma\,x+\beta\,p_y}{\hbar^2\omega^2},\qquad
{p_y}_R=\frac{\gamma\,y-\beta\,p_x}{\hbar^2\omega^2},
$$
valid for $\omega>0$.
\end{theorem}

\begin{proof}
In the case of a special Jordan algebra, an associator can be expressed as a double
commutator:
\begin{equation}
\label{eq:assoc-double}
[A,B,C]=\frac{1}{4}[B,[A,C]].
\end{equation}
Using only \eqref{eq:H-theta} and \eqref{eq:NC-Heisenberg}, a direct computation gives
$$
i[H_\theta,x]=\beta y+\alpha p_x,\qquad
i[H_\theta,y]=-\beta x+\alpha p_y,
$$
$$
i[H_\theta,p_x]=-\gamma x+\delta p_y,\qquad
i[H_\theta,p_y]=-\gamma y-\delta p_x .
$$
Since $\alpha\gamma-\beta\delta=\hbar^2\omega^2\neq0$ for $\omega>0$, these relations yield
$$
i[H_\theta,y_R]=y,\qquad i[H_\theta,{p_x}_R]=p_x,\qquad i[H_\theta,{p_y}_R]=p_y .
$$
Hence, by \eqref{eq:assoc-double},
$$
[4H_\theta,x,y_R]=[x,[H_\theta,y_R]]=-i[x,y]=\theta,
$$
$$
[4H_\theta,x,{p_x}_R]=-i[x,p_x]=\hbar,\qquad
[4H_\theta,y,{p_y}_R]=-i[y,p_y]=\hbar .
$$
We note that the partners ${}_L,{}_R$ are not unique. For $\omega=0$ one has
$\alpha\gamma-\beta\delta=0$, $y$ is not in the range of $i[H_\theta,\cdot\,]$ restricted to
$\mathrm{span}\{x,y,p_x,p_y\}$, and no linear partner $y_R$ exists.
\end{proof}

\section{Density Matrices and Pure States}

\subsection{Density Matrices in the JBW-Algebra Setting}

It turns out that one needs a notion of density matrix in the JBW-algebra $\mathcal{A}$
in order to describe the state of a quantum system. For the purpose of applications to
JBW-algebras, let us recall the following result from \cite{AS}, \cite{ASSS}.

\begin{theorem}[\cite{AS,ASSS}]
\label{thm:density-proj}
Let $\mathcal{A}$ be a commutative, order unit algebra that is the dual of a base norm
space $V$ such that the multiplication in $\mathcal{A}$ is separately $w^*$-continuous.
Then for each $a\in\mathcal{A}$ and each $\epsilon>0$ there are orthogonal projections
$p_1,p_2,\dots,p_n$ in the $w^*$-closed subalgebra generated by $a$ and $\mathbb{I}$,
and scalars $\lambda_1,\lambda_2,\dots,\lambda_n$ such that
$$
\left\|a-\sum_{i=1}^n\lambda_ip_i\right\|<\epsilon.
$$
The set of finite linear combinations of orthogonal projections is norm dense in
$\mathcal{A}^{**}$, where $\mathcal{A}^{**}$ denotes the bidual of $\mathcal{A}$, which
is a JBW-algebra.
\end{theorem}

\begin{definition}
A self-adjoint operator $\rho\in\mathcal{A}$ is a \emph{density matrix} if $\rho$ is
non-negative and $\mathrm{tr}(\rho)=1$.
\end{definition}

In what follows, for a density matrix $\rho$ we denote by $B:=\rho^{1/2}$ its positive
square root; since $\mathrm{tr}(B^2)=\mathrm{tr}(\rho)=1$, we have $B\in\hat{\mathscr{H}}_J$.

\begin{theorem}
\label{thm:expectation-values}
Assume that $\rho\in\mathcal{A}$ is a density matrix. Then the map
$\Phi_{\rho}:\mathcal{A}\rightarrow\mathbb{R}$ given by
$$
\Phi_{\rho}(A)=\mathrm{tr}(\rho\bullet A)
$$
is a family of expectation values.
\end{theorem}

\begin{proof}
In order to show this, we need to verify the following three properties of $\Phi_\rho$.
\begin{itemize}
\item[(i)] $\Phi_{\rho}(\mathbb{I})=1$.
\item[(ii)] $\Phi_{\rho}(A)$ is real for all $A\in\mathcal{A}$, and $\Phi_{\rho}(A)$ is
non-negative whenever $A$ is non-negative.
\item[(iii)] For any sequence $A_n\in\mathcal{A}$, if
$\lim_{n\rightarrow\infty}\|A_n\bullet\psi-A\bullet\psi\|=0$ holds for all
$\psi\in\hat{\mathscr{H}}_J$, then $\Phi_{\rho}(A_n)$ converges to $\Phi_{\rho}(A)$.
\end{itemize}
Note that (i) is immediate from the definition of a density matrix.
Let $B=\rho^{1/2}\in\hat{\mathscr{H}}_J$. By the 3-cyclicity \eqref{eq:trace-cyclicity}
and commutativity,
\begin{equation*}
\Phi_\rho(A)=\mathrm{tr}(B^2\bullet A)=\mathrm{tr}\bigl(B\bullet(B\bullet A)\bigr)
=\langle B,A\bullet B\rangle. \tag{$*$}
\end{equation*}
(ii) Since $\mathrm{tr}$ is real on self-adjoint elements, $\Phi_\rho(A)\in\mathbb{R}$. For
positivity, recall $U_B(A)=2B\bullet(B\bullet A)-B^2\bullet A$; by $(*)$,
$\mathrm{tr}(U_B(A))=2\Phi_\rho(A)-\Phi_\rho(A)=\Phi_\rho(A)$. Since $U_B$ is a positive map
\cite{HOS} and the trace is positive, $A\geq0$ implies $\Phi_\rho(A)\geq0$.
(iii) Choosing $\psi=B$ in the hypothesis, $A_n\bullet B\to A\bullet B$ in norm; by $(*)$ and
continuity of the inner product, $\Phi_\rho(A_n)=\langle B,A_n\bullet B\rangle\to
\langle B,A\bullet B\rangle=\Phi_\rho(A)$.
\end{proof}

For a normal state $B\in\hat{\mathscr{H}}_J$, the expectation value of an observable
$A\in\mathcal{A}$ is given by
$$
\Phi_B(A)=\langle B,A\bullet B\rangle=\mathrm{tr}(B\bullet(A\bullet B))
=\mathrm{tr}(B^2\bullet A)=\mathrm{tr}(\rho\bullet A),
$$
where $\rho=B^2$ is the corresponding density matrix.

\begin{definition}
A density matrix is called a \emph{pure state} if it can be written as a projection
operator in $\mathcal{A}$; otherwise it is called a \emph{mixed state}.
\end{definition}

\subsection{Pure States for the Non-commutative Landau Problem}

Now, we can revisit the non-commutative Landau problem. The Hilbert space for the
composite system is the Hilbert tensor product
$\hat{\mathscr{H}}_{J,+}\otimes\hat{\mathscr{H}}_{J,-}$ of the Hilbert spaces
$\hat{\mathscr{H}}_{J,+}$ and $\hat{\mathscr{H}}_{J,-}$ describing the subsystems.
Assuming $\rho$ is a density matrix on
$\hat{\mathscr{H}}_{J,+}\otimes\hat{\mathscr{H}}_{J,-}$, there exists a unique density
matrix $\rho_+$ on $\hat{\mathscr{H}}_{J,+}$ with the property that
$$
\mathrm{tr}(\rho_+\bullet A)=\mathrm{tr}(\rho\bullet(A\otimes\mathbb{I}_{\hat{\mathscr{H}}_{J,-}}))
$$
for all $A\in\mathcal{A}_+$. Similarly, there exists a unique
density matrix $\rho_-$ on $\hat{\mathscr{H}}_{J,-}$ with the property that
$$
\mathrm{tr}(\rho_-\bullet B)=\mathrm{tr}(\rho\bullet(\mathbb{I}_{\hat{\mathscr{H}}_{J,+}}\otimes B))
$$
for all $B\in\mathcal{A}_-$, where $\mathcal{A}_\pm$ denote the Jordan algebras of
observables of the two subsystems. Since $\mathrm{tr}(\rho_\pm\bullet A)=\mathrm{tr}(\rho_\pm A)$
for self-adjoint $A$, these are the usual partial traces, whose existence and uniqueness are
standard. We stress that the Jordan product of the composite system is the one induced by the
ambient associative algebra $\mathcal{B}(\mathscr{H}_+\otimes\mathscr{H}_-)$; a $JB$-algebra admitting a
suitable tensor product with $M_2(\mathbb{C})_{sa}$ is necessarily the self-adjoint part of a
$C^*$-algebra \cite{HancheOlsen}, so composite systems are in general not described by a Jordan
tensor product of the subsystem algebras. At the level of real Hilbert
spaces, the real tensor product of the self-adjoint Hilbert-Schmidt spaces coincides with the
space of self-adjoint Hilbert-Schmidt operators on $\mathscr{H}_+\otimes\mathscr{H}_-$. Since $H_\theta$ does not couple the two subsystems, product states are preserved by the
time evolution. In what follows we restrict attention to such product states
$\rho=\rho_+\otimes\rho_-$ (correlated states are, of course, also possible in general).

\begin{theorem}
\label{thm:pure-states}
A product density matrix $\rho=\rho_+\otimes\rho_-$ is a pure state if and only if
$\rho_\pm=P_{0,\pm}$ are rank-one projections. The corresponding state vector in
$\hat{\mathscr{H}}_{J,+}\otimes\hat{\mathscr{H}}_{J,-}$, i.e.\ the positive square root of
$\rho$, is $B=P_{0,+}\otimes P_{0,-}$; the vectors $\pm B$ represent the same state.
\end{theorem}

\begin{proof}
If $\rho$ is pure, then $\rho^2=\rho$, hence
$1=\mathrm{tr}(\rho)=\mathrm{tr}(\rho^2)=\mathrm{tr}(\rho_+^2)\,\mathrm{tr}(\rho_-^2)$.
For a density matrix $\sigma$ with eigenvalues $s_i\geq0$, $\sum_is_i=1$, one has
$\mathrm{tr}(\sigma^2)=\sum_is_i^2\leq1$, with equality if and only if exactly one
$s_i$ equals $1$. Hence both factors equal $1$ and $\rho_\pm$ are rank-one projections.
Conversely, the tensor product of two rank-one projections is a projection of trace one.
Finally, since $\rho$ is a projection, its positive square root is $\rho$ itself.
\end{proof}

\section{Jordan--Schr\"odinger Time Evolution}

The derivation of a consistent Schr\"{o}dinger time-evolution equation in a
non-associative setting, together with examples in Jordan algebras, is given in
Sections~6 and~7 of \cite{SchuppSzabo}. We are going to arrive at a time-evolution
equation for the Landau problem at our disposal in a JB-algebraic setting presented
above. For this purpose, we will explicitly reformulate the corresponding Hamiltonian
as a finite sum of commutators of observables. Split operators appear in proposals for
non-associative dynamics that replace commutators by an associator \cite{SchuppSzabo}.
Here our main advantage lies in obtaining the split operators explicitly. By split operators we mean self-adjoint operators $R_j,S_j$ with
$H=i\sum_j[R_j,S_j]$, which allow the commutator in the Liouville equation to be replaced by
associators. Since the construction uses only the associative product of the ambient
$C^*$-algebra, it may also be useful in other non-associative settings, e.g.\ for Lie triple
algebras \cite{LRH}; we do not pursue this here.

\subsection{Split Operators for the Non-commutative Landau Hamiltonian}

Now, we define new observables before stating our next theorem giving the time-evolution
equation for the state vectors of the Landau problem. Let
$$
X_+=C_{1_+}(A_+^{\dagger}+A_+), \qquad X_-=C_{1_-}(A_-^{\dagger}+A_-)
$$
be the position operators corresponding to the subsystems, and let
$$
P_{+}=iC_{2_+}(A_+^{\dagger}-A_+), \qquad P_{-}=iC_{2_-}(A_-^{\dagger}-A_-)
$$
be the corresponding momentum operators, where
$$
C_{1_+}=\sqrt{\frac{\hbar}{2\mu\Omega_+}}, \quad
C_{2_+}=\sqrt{\frac{\hbar\mu\Omega_+}{2}}, \quad
C_{1_-}=\sqrt{\frac{\hbar}{2\mu\Omega_-}}, \quad
C_{2_-}=\sqrt{\frac{\hbar\mu\Omega_-}{2}}.
$$
Here we use $\omega>0$, so that $\Omega_\pm>0$ (Section~4.3).

\subsection{Jordan--Schr\"odinger Equation and Exact Solution}

\begin{theorem}[cf.~\cite{SchuppSzabo}]
\label{thm:time-evolution}
Given a state vector $v_+\otimes v_-\in\hat{\mathscr{H}}_{J,+}\otimes\hat{\mathscr{H}}_{J,-}$
with $v_\pm$ in the domain of the generators described in Section~\ref{sec:fa} (in particular
$v_\pm\in\mathcal{F}_{J,\pm}$), the time evolution equation can be written as follows:
\begin{equation}
\label{eq:JS-equation}
\frac{d}{dt}(v_+\otimes v_-)=\frac{-4}{\hbar}\!\left(\!
\bigl([L_{S_{1_+}},L_{R_{1_+}}]+[L_{S_{2_+}},L_{R_{2_+}}]\bigr)v_+\otimes v_-
+v_+\otimes\bigl([L_{S_{1_-}},L_{R_{1_-}}]+[L_{S_{2_-}},L_{R_{2_-}}]\bigr)v_-
\!\right),
\end{equation}
where $H_+=i[R_{1_+},S_{1_+}]+i[R_{2_+},S_{2_+}]$,
$H_-=i[R_{1_-},S_{1_-}]+i[R_{2_-},S_{2_-}]$, and
$$
R_{1_+}=-\frac{\Omega_+}{4}(X_+A_++A_+^{\dagger}X_+),
\qquad
S_{1_+}=P_+A_++A_+^{\dagger}P_+,
$$
$$
R_{2_+}=-\frac{\Omega_+}{4}(P_+A_++A_+^{\dagger}P_+),
\qquad
S_{2_+}=2C_{1_+}A_+^{\dagger}A_+,
$$
$$
R_{1_-}=-\frac{\Omega_-}{4}(X_-A_-+A_-^{\dagger}X_-),
\qquad
S_{1_-}=P_-A_-+A_-^{\dagger}P_-,
$$
$$
R_{2_-}=-\frac{\Omega_-}{4}(P_-A_-+A_-^{\dagger}P_-),
\qquad
S_{2_-}=2C_{1_-}A_-^{\dagger}A_-.
$$
\end{theorem}

\begin{proof}
Let us first verify that $H_+=i[R_{1_+},S_{1_+}]+i[R_{2_+},S_{2_+}]$ and
$H_-=i[R_{1_-},S_{1_-}]+i[R_{2_-},S_{2_-}]$. We will only show the first equality,
since the two subsystems are independent identical copies. Note that
\begin{align}
\frac{-i}{C_{1_+}C_{2_+}}&[X_+A_++A_+^{\dagger}X_+,\,P_+A_++A_+^{\dagger}P_+]
\nonumber\\
&=[2A_+^{\dagger}A_++A_+^2+{A_+^{\dagger}}^2,\,{A_+^{\dagger}}^2-A_+^2]
\nonumber\\
&=[2A_+^{\dagger}A_+,{A_+^{\dagger}}^2]-[2A_+^{\dagger}A_+,A_+^2]+2[A_+^2,{A_+^{\dagger}}^2]
\nonumber\\
&=2\!\left(A_+^{\dagger}[A_+,{A_+^{\dagger}}^2]-[A_+^{\dagger},A_+^2]A_+
+[A_+,{A_+^{\dagger}}^2]A_++A_+[A_+,{A_+^{\dagger}}^2]\right)
\nonumber\\
&=4{A_+^{\dagger}}^2+4A_+^2+8A_+^{\dagger}A_++4\mathbb{I}.
\nonumber
\end{align}
Similarly,
\begin{align}
\frac{-i}{C_{2_+}}&[P_+A_++A_+^{\dagger}P_+,\,A_+^{\dagger}A_+]
\nonumber\\
&=[{A_+^{\dagger}}^2,A_+^{\dagger}A_+]-[A_+^2,A_+^{\dagger}A_+]
\nonumber\\
&=-[A_+^{\dagger}A_+,A_+^{\dagger}]A_+^{\dagger}-A_+^{\dagger}[A_+^{\dagger}A_+,A_+^{\dagger}]
+[A_+^{\dagger}A_+,A_+]A_++A_+[A_+^{\dagger}A_+,A_+]
\nonumber\\
&=-2{A_+^{\dagger}}^2-2A_+^2.
\nonumber
\end{align}
Combining these, we get the required result $H_+=i[R_{1_+},S_{1_+}]+i[R_{2_+},S_{2_+}]$.

\medskip

The time evolution equation for density matrices in a $C^*$-algebra is given by
\begin{equation}
\label{eq:Liouville}
\dot{\rho}=-\frac{i}{\hbar}[H,\rho],
\end{equation}
where $H$ denotes the Hamiltonian. In the case of Jordan algebras, inner derivations are
given in terms of the associator, which acts as a derivation on the second argument
\cite{HOS}. Assuming $H=i[R,S]$, where $R$ and $S$ are self-adjoint,
\eqref{eq:Liouville} transforms into the following equation:
$$
\dot{\rho}=\frac{-4}{\hbar}[R,\rho,S]=\frac{-4}{\hbar}[L_S,L_R]\rho,
$$
where $L_S$ and $L_R$ are left multiplication operators with $S$ and $R$, respectively.
This amounts to the Jordan-algebraic evolution equation for the density matrices. All
density matrices can be written in the form $\rho=B^2$, for some $B\in\mathcal{A}_2$.
Taking the derivative of both sides:
$$
\dot{\rho}=2B\bullet\dot{B}.
$$
On the other hand, we also have
$$
\dot{\rho}=\frac{-4}{\hbar}[R,\rho,S]=\frac{-4}{\hbar}[R,B^2,S]
=\frac{-4}{\hbar}\cdot 2B\bullet[R,B,S].
$$
This relation alone does not determine $\dot B$, since $L_B$ need not be injective
(e.g.\ for pure states). However, $D:=-\frac{4}{\hbar}[L_S,L_R]$ coincides with
$-\frac{i}{\hbar}[H,\cdot\,]$, so that $e^{tD}$ acts by unitary conjugation,
$e^{tD}v=U_tvU_t^*$ with $U_t=e^{-iHt/\hbar}$ (made precise below). Hence
$B(t):=e^{tD}B(0)$ is positive and satisfies $B(t)^2=U_t\rho(0)U_t^*=\rho(t)$; by uniqueness
of the positive square root, $B(t)=\rho(t)^{1/2}$. Thus the positive square root of $\rho$
evolves according to
$$
\dot{B}=\frac{-4}{\hbar}[R,B,S]=\frac{-4}{\hbar}[L_S,L_R]B.
$$
Let us remark that this equation can be extended to the whole space $\hat{\mathscr{H}}_J$
by
$$
\dot{v}=\frac{-4}{\hbar}[R,v,S]=\frac{-4}{\hbar}[L_S,L_R]v.
$$
Next, let $v_+\otimes v_-$ be an element of
$\hat{\mathscr{H}}_{J,+}\otimes\hat{\mathscr{H}}_{J,-}$ corresponding to the density
matrix $\rho_+\otimes\rho_-$, and let
$H_{\theta}=H_+\otimes\mathbb{I}_{\hat{\mathscr{H}}_{J,-}}
+\mathbb{I}_{\hat{\mathscr{H}}_{J,+}}\otimes H_-$ be the Hamiltonian, where
$H_+=i[R_{1_+},S_{1_+}]+i[R_{2_+},S_{2_+}]$ and
$H_-=i[R_{1_-},S_{1_-}]+i[R_{2_-},S_{2_-}]$. Based on these, we may deduce that
\begin{align*}
\frac{d}{dt}(\rho_+\otimes\rho_-)
&=-\frac{i}{\hbar}\bigl[H_+\otimes\mathbb{I}_{\hat{\mathscr{H}}_{J,-}}
+\mathbb{I}_{\hat{\mathscr{H}}_{J,+}}\otimes H_-,\,\rho_+\otimes\rho_-\bigr]\\
&=-\frac{i}{\hbar}\bigl(
[H_+\otimes\mathbb{I}_{\hat{\mathscr{H}}_{J,-}},\rho_+\otimes\rho_-]
+[\mathbb{I}_{\hat{\mathscr{H}}_{J,+}}\otimes H_-,\rho_+\otimes\rho_-]\bigr)\\
&=-\frac{i}{\hbar}\bigl(
([H_+,\rho_+]\otimes\rho_-)+(\rho_+\otimes[H_-,\rho_-])
\bigr).
\end{align*}
Using the discussion above, we may replace $-i[H,\rho]$ with
$-4[L_{S_1},L_{R_1}]v-4[L_{S_2},L_{R_2}]v$ to arrive at the formula
\eqref{eq:JS-equation}.

Finally, the solution to this Jordan--Schr\"{o}dinger equation can be written as
\begin{equation}
\label{eq:exact-solution}
v_+(t)\otimes v_-(t)=
e^{\frac{-4t}{\hbar}([L_{S_{1_+}},L_{R_{1_+}}]+[L_{S_{2_+}},L_{R_{2_+}}])}v_+(0)
\otimes
e^{\frac{-4t}{\hbar}([L_{S_{1_-}},L_{R_{1_-}}]+[L_{S_{2_-}},L_{R_{2_-}}])}v_-(0).
\end{equation}
The exponentials in \eqref{eq:exact-solution} are defined as follows. By Stone's theorem,
$U_t^\pm=e^{-iH_\pm t/\hbar}$ are strongly continuous unitary groups, and the maps
$V_t^\pm(v)=U_t^\pm v\,(U_t^\pm)^*$ form strongly continuous one-parameter groups of orthogonal
transformations of $\hat{\mathscr{H}}_{J,\pm}$, since they preserve self-adjointness and the
scalar product $\mathrm{tr}(v\bullet w)$. As shown in Section~\ref{sec:fa} and Appendix~\ref{app:fa}, their generators are the closures of
$-\frac{4}{\hbar}\bigl([L_{S_{1_\pm}},L_{R_{1_\pm}}]+[L_{S_{2_\pm}},L_{R_{2_\pm}}]\bigr)
=-\frac{i}{\hbar}[H_\pm,\cdot\,]$ defined on $\mathcal{F}_{J,\pm}$, and \eqref{eq:exact-solution} means
$v_\pm(t)=U_t^\pm v_\pm(0)(U_t^\pm)^*$; no product formula is needed.
\end{proof}

\begin{remark}
In the number basis $\{|n\rangle\langle m|\}$ of the $\pm$ mode, the solution reads
$$
v_{nm}(t)=e^{-i\Omega_\pm(n-m)t}\,v_{nm}(0).
$$
Equivalently, in the real basis \eqref{eq:basis-HJ}, the elements $|n\rangle\langle n|$ are
stationary, while each pair $\frac12(|n\rangle\langle m|+|m\rangle\langle n|)$,
$\frac{i}{2}(|n\rangle\langle m|-|m\rangle\langle n|)$ is rotated by the angle
$\Omega_\pm(n-m)t$. The evolution preserves the trace, since
$\mathrm{tr}([R,\rho,S])=0$ for $\rho\in\mathcal{F}_{J,\pm}$ by the 3-cyclicity
\eqref{eq:trace-cyclicity}, and the norm,
since $L_R$ and $L_S$ are symmetric with respect to $\langle v,w\rangle=\mathrm{tr}(v\bullet w)$,
so that $[L_S,L_R]$ is antisymmetric. For $\theta\to0$ one has
$\Omega_\pm\to\sqrt{\omega^2+\omega_L^2/4}\pm\omega_L/2$, the familiar frequencies of the
commutative Landau problem with a harmonic potential.
\end{remark}

\section{Concluding Remarks}

In this paper we have provided a self-contained review of the Jordan algebraic,
non-associative formulation of quantum mechanics and its application to the
non-commutative Landau problem coupled to a harmonic potential, complementing and
extending the pioneer construction of \cite{SchuppSzabo}. The paper has covered the
relevant operator algebra background --- $C^*$-algebras, von Neumann algebras,
JB-algebras, and JBW-algebras --- and presented the axiomatic formulation of quantum
mechanics in terms of the Jordan algebra of observables, including the GNS-type Hilbert
space construction based on 3-cyclic traces.

In Theorem~\ref{thm:associator}, the non-commutative parameter $\theta$ of the non-commutative plane, motivated by
the non-commutativity of the guiding-center coordinates, is written in terms of an associator in the
Jordan algebra of observables. Secondly, a general understanding for the expectation
value of an observable is provided in terms of the Jordan algebraic setting
(Theorem~\ref{thm:expectation-values}). In addition, the pure product states are characterized
for this specific problem (Theorem~\ref{thm:pure-states}). Lastly, a complete derivation
of the Jordan algebraic time evolution equation for the state vectors is given in the
constructed Hilbert space (Theorem~\ref{thm:time-evolution}), and its solution is
also presented in terms of exponentials with operator arguments \eqref{eq:exact-solution},
which we identified with unitary conjugation and wrote out explicitly in the number basis. Decompositions of the Hamiltonian and the unbounded
operators for the Landau problem are explicitly obtained. This may lead to other potential
applications pertaining to non-associative problems, such as applications to Lie triple
algebras, which turn out to be a generalization of Jordan algebras. In conclusion, the paper provides a non-associative, infinite dimensional,
probabilistic description of a concrete physical model, the non-commutative Landau problem
with a harmonic potential, based on the Hilbert space realization of the states given in
Section~3. The probabilistic nature
of our construction further manifests itself in the density matrix formalism presented
above. This interpretation is crucial for the success of our development since the nature
of quantum mechanics depends on probabilistic considerations. Our approach may potentially
ignite further interesting constructions of non-associative algebras through
nonconventional settings. It would be beneficial to further study such possibilities in
the future from different aspects.

\section*{Acknowledgements}

We are grateful to Professors R.\,J.\ Szabo, P.\ Schupp, and the anonymous referees
for constructive comments and suggestions. T.D.\ thanks The Turkish Academy of Sciences
(T\"{U}BA) for partial support.

\section*{Data Availability}

Data sharing is not applicable to this article.

\appendix
\section{Domains and Self-Adjointness of the Unbounded Operators}
\label{app:fa}

\emph{Domain.} Since $A_\pm,A_\pm^\dagger$ satisfy the canonical commutation relations,
$\mathscr{H}\cong\mathscr{H}_+\otimes\mathscr{H}_-$, where $\mathscr{H}_\pm$ is the Fock space of
the mode $A_\pm$, and $\mathcal{D}=\mathcal{D}_+\otimes_{\rm alg}\mathcal{D}_-$ with
$\mathcal{D}_\pm=\mathrm{span}\{|n\rangle_\pm\}$. The elements of $\mathcal{D}$ are Gaussians
multiplied by polynomials, so that $\mathcal{D}\subset\mathcal{S}(\mathbb{R}^2)$.

\emph{Essential self-adjointness.} Each operator $T$ listed in Section~\ref{sec:fa}, except $A_\pm$, is symmetric
on $\mathcal{D}$, and every $\psi\in\mathcal{D}$ is an analytic vector for $T$. Indeed, let $T$
consist of at most $K$ monomials of degree at most two in $A_\pm,A_\pm^\dagger$ with coefficients
bounded in modulus by $c$. On vectors of total particle number at most $m$, each monomial has norm
at most $m+2$ and raises the particle number by at most two. Hence, for $\psi\in\mathcal{D}$ of
total particle number at most $n$,
$$
\|T^k\psi\|\leq(Kc)^k\prod_{j=0}^{k-1}(n+2j+2)\,\|\psi\|
\leq(2Kc)^k\,\frac{(m+k)!}{m!}\,\|\psi\|,\qquad m=\lceil n/2\rceil,
$$
so that $\sum_k\|T^k\psi\|\,t^k/k!<\infty$ for $|t|<(2Kc)^{-1}$. By Nelson's analytic vector
theorem \cite{Nelson,ReedSimon2}, $T$ is essentially self-adjoint on $\mathcal{D}$. In particular,
the closure of $H_\theta$ coincides with $H_+\otimes\mathbb{I}+\mathbb{I}\otimes H_-$.

\emph{Jordan products.} $\mathcal{F}_{J,\pm}$ is the space of self-adjoint finite-rank operators with range in $\mathcal{D}_\pm$. Let $T$ be one of the single-mode operators $H_\pm$, $N_\pm$, $X_\pm$, $P_\pm$, $R_{i_\pm}$, $S_{i_\pm}$, which are symmetric on $\mathcal{D}_\pm$ and map it into itself. If $v=\sum_i|\varphi_i\rangle\langle\psi_i|\in\mathcal{F}_{J,\pm}$,
then $Tv=\sum_i|T\varphi_i\rangle\langle\psi_i|$ and, by the symmetry of $T$ on $\mathcal{D}_\pm$,
$vT=\sum_i|\varphi_i\rangle\langle T\psi_i|$ are again of finite rank with range in
$\mathcal{D}_\pm$, and $(Tv)^*=vT$. Hence $T\bullet v\in\mathcal{F}_{J,\pm}$, so that
$\mathcal{F}_{J,\pm}$ is invariant under all left multiplications $L_T$ and their products and
commutators.

\emph{Time evolution.} The unitary groups $U_t^\pm=e^{-iH_\pm t/\hbar}$ act diagonally on number
states and leave $\mathcal{D}_\pm$ invariant. Hence $V_t^\pm(v)=U_t^\pm v\,(U_t^\pm)^*$ leaves
$\mathcal{F}_{J,\pm}$ invariant and is differentiable on it, with derivative
$-\frac{i}{\hbar}[H_\pm,v]=-\frac{4}{\hbar}\bigl([L_{S_{1_\pm}},L_{R_{1_\pm}}]
+[L_{S_{2_\pm}},L_{R_{2_\pm}}]\bigr)v$. On the complex Hilbert space of all Hilbert-Schmidt
operators, $V_t^\pm$ is a strongly continuous unitary group, and the dense invariant subspace
spanned by the operators $|\varphi\rangle\langle\psi|$ with $\varphi,\psi\in\mathcal{D}_\pm$ is a
core for its generator \cite[Thm.~VIII.10]{ReedSimon1}. Since $V_t^\pm$ commutes with
$v\mapsto v^*$, the same holds for the restriction to $\hat{\mathscr{H}}_{J,\pm}$ and
$\mathcal{F}_{J,\pm}$. Thus the generator of $V_t^\pm$ is the closure of the right-hand side of
\eqref{eq:JS-equation} defined on $\mathcal{F}_{J,\pm}$, and \eqref{eq:JS-equation} holds for all
$v_\pm$ in its domain.


\end{document}